\documentclass[aps,amssymb,amsmath,amsfonts,superscriptaddress,showkeys]{revtex4-1}

\usepackage[bookmarksnumbered, colorlinks, plainpages]{hyperref}
\usepackage{epsf,epsfig}
\usepackage{amssymb}
\usepackage{stmaryrd}
\usepackage{amsmath}
\usepackage{MnSymbol}
\usepackage[caption=false]{subfig}
\usepackage{color,soul}
\usepackage{xcolor}
\usepackage{pgf,tikz}
\usepackage{graphicx}
\usepackage{dcolumn}
\usepackage{bm,bbm}
\usepackage{color,soul}
\usepackage{xcolor}
\usepackage{import}
\usepackage{csquotes}
\usepackage{amsthm}
\usepackage[toc,page]{appendix}
\graphicspath{{./fig/}}
\usepackage{enumerate}
\usepackage{listings}
\usepackage{soul}
\usepackage[]{xcolor}

\newtheorem{theorem}{Theorem}
\newtheorem{lemma}[theorem]{Lemma}
\newtheorem{proposition}[theorem]{Proposition}
\newtheorem{corollary}[theorem]{Corollary}
\newtheorem{remark}[theorem]{Remark}

\newcommand{\pref}[1]{Proposition~\ref{#1}}

\newcommand{\dd}{\ensuremath{\mathrm{d}}}
\newcommand{\1}{{\rm 1\hspace{-0.9mm}l}}
\newcommand{\matrizesp}{\ensuremath{U_c}}

\newcommand{\sub}{\substack}
\DeclareMathOperator{\adiag}{adiag}
\DeclareMathOperator{\diag}{diag}
\newtheorem{conj}{Conjecture}[section]

\newcommand{\unicha}{\ensuremath{\mathcal{U}^{\mathrm{Q}}_3}{}}

\newcommand{\lindblad}{\ensuremath{\mathcal{A}_3}{}}

\newcommand{\ket}[1]{|#1\rangle}
\newcommand{\bra}[1]{\langle #1|}
\newcommand{\Ket}[1]{|#1\rangle\rangle}
\newcommand{\Bra}[1]{\langle\langle #1|}
\newcommand{\project}[1]{\ket{#1}\bra{#1}}

\newcommand{\outprod}[2]{\ket{#1}\bra{#2}}
\newcommand{\Id}{{\mathbb I}}
\newcommand{\e}{{\mathrm e}}
\renewcommand{\c}[1]{\mathcal{#1}}

\newcommand{\tr}[1]{\mbox{Tr} #1}

\newcommand{\pr}{\Pi}

\renewcommand{\overrightarrow}{\vec}

\begin{document}
\title{Log-Convex set of Lindblad semigroups acting on $N$-level system}
\author{Fereshte Shahbeigi}
\affiliation{Department of Physics, Ferdowsi University of
Mashhad, Mashhad, Iran}
\affiliation{Department of Physics, Sharif University of Technology, Tehran, Iran}
\author{David Amaro-Alcal\'a}
\affiliation{Instituto de F\'isica, Universidad Nacional
  Aut\'onoma de M\'exico, M\'exico D.F. 01000, M\'exico}
\author{Zbigniew Pucha{\l}a}
\affiliation{Institute of Theoretical and Applied Informatics, Polish Academy
of Sciences, ulica Ba{\l}tycka 5, 44-100 Gliwice, Poland}
\affiliation{Faculty of Physics, Astronomy and Applied Computer Science,
Jagiellonian University, ul. {\L}ojasiewicza 11,  30-348 Krak{\'o}w, Poland}
\author{Karol {\.Z}yczkowski}
\affiliation{Faculty of Physics, Astronomy and Applied Computer Science,
Jagiellonian University, ul. {\L}ojasiewicza 11,  30-348 Krak{\'o}w, Poland}
   \affiliation{Center for Theoretical Physics, Polish Academy of Sciences, Al. Lotnik{\'o}w
32/46, 02-668 Warszawa, Poland}
 \affiliation{National Quantum Information Centre (KCIK), University of Gda{\'n}sk,
  81-824 Sopot, Poland}

\date{March 26, 2020}

\begin{abstract}
We analyze the set ${\cal A}_N^Q$ of mixed unitary channels represented in the
Weyl basis and accessible by a Lindblad semigroup acting on an $N$-level quantum
system. General necessary and sufficient conditions for a mixed Weyl quantum
channel of an arbitrary dimension  to be accessible by a semigroup are
established. The set ${\cal A}_N^Q$  is shown to be log--convex and star-shaped
with respect to the completely depolarizing  channel.
A decoherence supermap acting in the space of Lindblad operators
transforms them into the space of Kolmogorov generators of
classical semigroups. We show that for mixed Weyl channels
the hyper-decoherence commutes with the dynamics,
so that decohering a quantum accessible channel
we obtain a bistochastic matrix form the set ${\cal A}_N^C$
of classical maps accessible by a semigroup.
Focusing on $3$-level systems we investigate the geometry
of the sets of quantum accessible maps, its classical counterpart
and the support of their spectra.
We  demonstrate that the set  ${\cal A}_3^Q$
is not included in the set ${\cal U}^Q_3$
of quantum unistochastic channels,
although an analogous relation holds for $N=2$.
The set of transition matrices obtained
by hyper-decoherence of unistochastic channels of order $N\ge 3$
is shown to be larger than the set of unistochastic matrices of this order,
and yields a motivation to introduce the larger sets of $k$-unistochastic
matrices.
\end{abstract}

\maketitle

\section{Introduction}
Open quantum systems have attracted a lot of interest from
different perspectives \cite{BP02,RH11}. Two major
methods to describe the dynamics of a quantum system
interacting with an environment are frequently used.
The first one takes place in continuous time,
as one deals with differential equations of motion
governing the time evolution of the system.
The most general form of such a non-unitary quantum dynamics
was introduced by Gorini, Kossakowski, Sudarshan, and independently by
Lindblad \cite{GKS76,Li76}. For a historical account on these
fundamental results, see the recent review \cite{CP17}.

In the second, stroboscopic approach, time changes in a discrete way
and the evolution of a quantum system is described by
quantum operations: linear  completely positive maps
which preserve the trace of the density matrix.
Let
$\Phi:\mathrm{M}_N\rightarrow\mathrm{M}_N$ be a
quantum operation acting on an $N$-dimensional system
described in the complex Hilbert space $\mathcal{H}_N$.
The action of this map on an
$N$-level density matrix $\rho$ in terms of Kraus operators, $K_j$,
reads,
 $\Phi(\rho)=\sum K_j\rho K_j^\dagger$.  In an
alternative approach, one can represent  a channel by the corresponding
superoperator $\Phi$ acting on an extended Hilbert
space, $\mathcal{H}_N\otimes\mathcal{H}_N$.
The superoperator, denoted by the same symbol $\Phi$ as the map it represents,
 can be determined  by its entries in a product basis,
  $\Phi_{\stackrel{\scriptstyle m \mu}{n\nu}} = \bra{m \mu}
\Phi \ket{n  \nu}=\sum_j
(K_j)_{mn}(\overline{K}_j)_{\mu\nu}$.
It is known that a given quantum channel may have many corresponding
sets of Kraus operators, while the superoperator
associated with a channel is unique.

  Although the superoperator  $\Phi$ is not necessarily Hermitian,
  \emph{reshuffling} its entries results in a positive
  semidefinite dynamical matrix \cite{SMR61}, also known as the Choi matrix  \cite{Cho75a}
  of the channel,
   $D_{\Phi} = \Phi^{R}$, where
$X_{\stackrel{\scriptstyle m \mu}{n\nu}}^R
=X_{\stackrel{\scriptstyle m n}{\mu \nu}}$ -- see  \cite{unistochasticos3}.
It is customary to refer to the $N^2$ eigenvalues of $\Phi$ as
the channel eigenvalues.
Due to Hermiticity and trace preserving conditions,
the eigenvalues of a channel
are either real or appear in complex conjugate pairs.
There exists at least one leading eigenvalue equal to unity,
which  corresponds to the invariant state \cite{BZ17}.
 It is interesting to study relations between both approaches \cite{CW13}
 and investigate which discrete channel $\Phi$ can be accessed by a
 continuous dynamics in the Lindblad form.
 This general question, sometimes called \emph{embedding} problem
 or \emph{Markovianity} problem, is known to be hard for a large system size  \cite{CEW12}.
 A necessary condition is the \emph{divisibility} property \cite{WC08,WECC08,FPMZ17},
 so that the map $\Phi$ can be represented as a composition of two
 other quantum maps, $\Phi=\Psi_2\Psi_1$,
 which are completely positive and trace preserving
 such that one of them is not unitary.
 There exist non-divisible quantum maps,
 which correspond to non-Markovian dynamics \cite{RHP14,CW15}.

 In the simplest case of  $N=2$
 and one-qubit unital maps,
 which are unitary equivalent to Pauli channels,
  $\Phi_p=\sum p_i \sigma_i \otimes \overline{\sigma}_i$,
 necessary and  sufficient conditions for
 accessibility by a dynamical semigroup were
 established \cite{WECC08,DZP18,PRZ19}.
 If all eigenvalues $\lambda_i$ of $\Phi_p$ are real and positive,
 the map $\Phi$
 is shown to be accessible by a semigroup if and only if
 the relation for subleading eigenvalues,  $\lambda_i\geq\lambda_j\lambda_k$,
 holds for
 any choice of three different indices $i$, $j$ and $k$ from the set
  $\{1,2,3\}$, as the leading eigenvalue reads
 $\lambda_{0}=1$.
 These relations are equivalent to the following inequalities
 for the weights in the convex combination of Pauli matrices,
   $p_0p_i \geq p_jp_k$,
 where the weight $p_0$ standing by $\sigma_0=\Id_2$
 is assumed to be the largest.
 Furthermore, it was shown that the set ${\mathcal A}^Q_2$
 of quantum channels accessible by a semigroup,
 forms a fourth part of the set ${\cal U}_2^Q$
 of quantum unistochastic channels \cite{PRZ19},
 determined by unitary matrices of size  $N^2=4$
 which couple the principal system with the environment
 initially in the maximally mixed state \cite{unistochasticos3,unistochastic}.
Further geometric properties of the set of one qubit Pauli maps
were recently investigated in \cite{Si19}.

The aim of the present work is to generalize earlier results
on the structure of the set  ${\cal A}_2^Q$
of one-qubit accessible channels \cite{PRZ19,DZP18}
for maps acting on  $N$-level systems.
In this work we concentrate on a class of
mixed unitary channels (also called random external fields \cite{AL07}),
defined as convex combinations of rotations
by Weyl matrices, which can be considered as
unitary generalizations of the Pauli matrices.
The set ${\cal W}_N$ of Weyl mixed unitary channels of size $N$
was recently investigated  to study the
degree  of non-Markovianity  \cite{CW15}.
Here we find necessary and sufficient conditions for
a quantum Weyl channel acting on an $N$ dimensional system
to be  accessible by a Lindblad  semigroup.
Furthermore, we  investigate  the set
 ${\cal A}_N^Q$ of Weyl channels accessible by a semigroup
 and demonstrate that it is log-convex and star-shaped.
This property is also inherited by the
corresponding set  ${\cal A}_N^C$
of circulant bistochastic matrices accessible by classical
semigroups.
As the relative volumes of different classes of quantum maps
were investigated \cite{SWZ08,JSP19,Si19},
we analyze the ratio of the volume of the set ${\cal A}_N^Q$
of accessible maps to the total volume of the set
 ${\cal W}_N$ of Weyl channels.

 This paper is organized as follows.
 In Section~\ref{sec:one} we fix the notation,
  introduce the basis of unitary Weyl  matrices
 and define channels accessible by a  Lindblad semigroup.
 Furthermore, we recall the definitions
 of stochastic, bistochastic and unistochastic quantum channels,
 and discuss the corresponding classical transition
 matrices.
 In Section~\ref{sec:two}, we establish a condition for
the Weyl  channels  to be
 accessible by a semigroup.
  Geometric properties of the set
 ${\cal A}_N^Q$ of accessible quantum channels are
  analyzed in Section~\ref{sec:three}.

 In  Section~\ref{sec:six} we investigate properties
of classical transition matrices generated
from the Weyl channels by  hyper-decoherence,
and demonstrate that these bistochastic matrices are circulant.
As decoherence acts in the space of quantum states \cite{Sch19},
an analogous process acting in the larger  space of quantum maps and
described by a supermap is sometimes called hyper-decoherence.
We show that  the semigroup dynamics
commutes with hyper-decoherence for mixed Weyl channels,
so that any accessible quantum map subjected to decoherence leads to
a bistochastic matrix accessible by a classical semigroup.
In mathematical literature such Markov stochastic matrices are also
called \emph{embeddable} \cite{Ru62,Da10}
and this property,  related to  classical analogue of divisibility of a channel  \cite{BC16},
was also studied in quantum context \cite{LKM17}.

Explicit conditions for a circulant bistochastic matrix,
obtained  by hyper-decoherence of mixed Weyl channels,
to be accessible by classical semigroup are derived
and geometry of the set  $\c A_N^C$ of these matrices is analyzed.
 In Section~\ref{sec:four},  we focus on $3$-level systems,
  called \emph{qutrits}, to emphasize
 prominent  difference with respect to the simpler
case of maps  acting on $2$-level systems.
  Possible relations between the sets of
 accessible channels,  unistochastic channels,
 and quantum channels,
 such that the corresponding classical transition matrix
 is unistochastic  are analyzed
 in Section~\ref{sec:five}.
  In particular, we demonstrate that for $N=3$
 there exists a quantum map accessible by a semigroup,
  which is not unistochastic. The work is concluded in Section   \ref{Sec7},
 in which we summarize results achieved and present a list of open questions.
 Some properties of unistochastic matrices of order $N=3$ and unistochastic channels
 acting in this dimension are discussed in Appendices.

 \section{Setting the scene}
 \label{sec:one}
   \subsection{Unitary Weyl matrices}
 \label{sec:Weyl}

In this Section we review necessary results on discrete
quantum maps and continuous dynamics  of the Lindblad form.
We will be mainly concerned  with  mixed unitary channels,
defined by a convex combination of unitary transformations,
\begin{equation}
\label{mixeduni}
\Phi_p \left(\rho\right)=\sum_{n=1}^M p_nV_n\rho V_n^\dagger.
\end{equation}
Here $\{V_n\}_{n=1}^M$ denotes an arbitrary collection of $M$ unitary
matrices of order $N$, while
${\vec p}=(p_1,\dots, p_M)^T$ represents a probability vector.
By  definition, any mixed unitary channel is unital, i.e.
$\Phi \left(\Id\right)=\Id$.
In this work we focus our attention on the set  ${\cal W}_N$  of
mixed unitary channels
in the form of Eq. \eqref{mixeduni},
where matrices $V_n$ form the basis of
\emph{Weyl} unitary matrices,
\begin{eqnarray}
\label{weyl1}
U_{kl}=X^kZ^l, \ \ k,l=0,\dots,N-1.
\end{eqnarray}
Here $X$ denotes the shift operator, $X\ket i=\ket{i\oplus1}$,
and $Z$ is a diagonal unitary matrix,
$Z=\diag\{1,\omega_{_N},\omega_{_N}^2,\dots,\omega_{_N}^{N-1}\}$
with  $\omega_{_N}=\e^{i\frac{2\pi}{N}}$.
The notation $i\oplus j$, refers to addition $i+j$ modulo $N$.
For convenience, we often  use a single index
  $\mu=0,\dots,N^2-1$, instead of two, $\mu = Nk+l$.

 This set of unitary matrices, applied in context of quantum
 theory  by Weyl \cite{We27},
 and later popularized by Schwinger \cite{Sch60},
was invented by Sylvester \cite{Sy09}  already in 1867.
These matrices form $N$-dimensional unitary representation
 of the elements of the Weyl-Heisenberg group
 and they are sometimes called \emph{clock--and--shift}
matrices, which  refers to the operators $Z$ and $X$, respectively.

With the exception of the identity matrix, $U_{00}= \Id_N$,
all remaining Weyl matrices are traceless.
Full set of $N^2$ matrices $U_{\mu}$  forms an orthogonal basis
composed of unitary matrices of size $N$,
which can be considered as a unitary generalization of Pauli
matrices of size $N=2$. It is known that
any one-qubit bistochastic channel $\Phi_B$ is unitarily equivalent to a
Pauli channel,
$\Phi_B(\rho)= U_2\bigl(\Phi_p(U_1\rho U_1^{\dagger})\bigr)U_2^{\dagger}$,
which is not the case for Weyl channels acting in higher dimensions.
Furthermore, for one-qubit systems all unital channels  are mixed
unitary, which is not true in higher dimensions  \cite{LS93}.

\medskip

It is easy to see that the Weyl unitary matrices of an arbitrary order $N$
enjoy the following property,

a) $U_{kl}U_{k^\prime l^\prime}=\omega_{_N}^{lk^\prime}
U_{k\oplus k^\prime,l\oplus l^\prime}$,
 which will be crucial in the further analysis.
 In turn it implies,

b) Commutativity up to a phase,  $U_{kl}U_{k^\prime l^\prime}=
\omega_{_N}^{lk^\prime- kl^\prime}
U_{k^\prime l^\prime}U_{kl}$,

c) Reflection symmetry, $U_{k,l}U_{-k\oplus N,-l\oplus N}=
\omega_{_N} ^{-lk}\; \Id_N$.

   \subsection{Mixed unitary channels accessible by a Lindblad semigroup}
  \label{lindbladchar}
Through celebrated GKLS theory \cite{GKS76,Li76},
the time evolution of a quantum system experiencing
Markovian dynamics can be described by a Lindblad generator $\mathcal{L}$
 and represented
 as $\rho(t)=\e^{\mathcal{L}t}[\rho(0)]=\Lambda_t[\rho(0)]$.
The action of the generator $\mathcal{L}$ on a quantum state $\rho$ of size $N$
can be written in terms of at most $N^2-1$ jump operators $L_j$,
\begin{eqnarray}
\label{lindblad1}
\mathcal{L}(\rho)=\sum_{j=1}^{N^2-1}\Bigl( L_j \rho L_j^{\dagger} -\frac{1}{2} L_j^{\dagger} L_j \rho -\frac{1}{2} \rho L_j^{\dagger} L_j \Bigr).
\end{eqnarray}
The corresponding superoperator  $\mathcal{L}$
can be represented as a matrix acting in the extended space  ${\cal H}_N \otimes {\cal H}_N$,
\begin{eqnarray}\label{lindblad2}
\mathcal{L}=\sum_{j=1}^{N^2-1} L_j \otimes \overline{L}_j	- \frac{1}{2} \sum_j L_j^\dagger L_j \otimes \Id_N- \frac{1}{2} \Id_N \otimes\sum_j  L_j^T  \overline{L}_j.
\end{eqnarray}

With these preliminaries in mind, let $\Phi_{\vec{p}}$ be a mixed
unitary quantum channel written in terms of the Weyl unitary
matrices:
\begin{eqnarray}
\label{phi}
\Phi_{\vec{p}}=\sum_{k,l=0}^{N-1} p_{kl}U_{kl}\otimes\overline{U}_{kl}=\sum_{\mu=0}^{N^2-1}p_\mu U_\mu\otimes\overline{U}_\mu.
\end{eqnarray}
Probabilities in this convex combination form an $(N^2-1)$-dimensional simplex
 $\vec{p}\in\Delta_{N^2-1}\subset\mathbb{R}^{N^2-1}$.
In this paper we find out,  for which $\vec{p}$
the map is semigroup accessible,
$\Phi_{\vec{p}} \in {\cal A}_N^Q$,
so that there exists a Lindblad generator $\mathcal{L}$
and $t>0$ such that $\e^{\mathcal{L}t}=\Phi_{\vec{p}}$.

A quantum channel $\Phi$ belongs to a dynamical semigroup if and only if there is
a Lindblad generator $\mathcal{L}$ such that $\Phi=\e^\mathcal{L}$, where $\mathcal{L}$ fulfils the following three conditions \cite{WECC08}:
\begin{enumerate}[i)]
 \item  Hermiticity preserving, $\mathcal{L}[Y^\dagger]=\mathcal{L}[Y]^\dagger$
 for any operator $Y$.
 This in turn results in the following condition
  for entries  of the Lindblad generator,
  $\mathcal{L}_{\stackrel{\scriptstyle m \mu}{n \nu}}=\overline{\mathcal{L}}_{\stackrel{\scriptstyle \mu m}{\nu n}}$.
 \item  Trace preserving, $\tr\left(\mathcal{L}[Y]\right)=0$ for any $Y$, which implies,
 $\forall n,\nu:\ \ \sum_m\mathcal{L}_{\stackrel{\scriptstyle m m}{n \nu}}=0$.
 \item  Conditionally completely positive, $(\Id_{N^2}-\project{\psi_+})(\mathcal{L}\otimes\Id_N)[\project{\psi_+}](\Id_{N^2}-\project{\psi_+})\geq0$,
  where $\ket{\psi_+}=\frac{1}{\sqrt{N}}\sum\ket{ii}$ is the maximally entangled state acting on the composite space
${\cal H}_A  \otimes {\cal H}_B$.
 This condition is equivalent to $(\Id_{N^2}-\project{\psi_+}) \mathcal{L}^R(\Id_{N^2}-\project{\psi_+})\geq0$ in
 which $\mathcal{L}^R$ is the reshuffled form of  $\mathcal{L}$.
\end{enumerate}

However, as  logarithm of a matrix is usually not unique
\cite{G59,C66}, it is not  straightforward to
verify the last condition for an arbitrary map $\Phi$.

 \subsection{Bistochastic quantum maps and classical bistochastic matrices}
  \label{bisto}

Any completely positive map $\Phi$ acting on a
quantum system of size $N$ can be represented
by the corresponding dynamical matrix,
determined by the Choi--Jamio{\l}kowski isomorphism,
$D_{\Phi}=N( \Phi \otimes {\mathbb I}_N)|\psi_+\rangle\langle \psi_+|$.
As discussed in the previous Section,
the dynamical matrix is related to the superoperator $\Phi$
by a particular reordering of the entries of the matrix, called
reshuffling, $D_{\Phi} = \Phi^{R}$.

A coarse-graining channel sends any density matrix $\rho$ into its diagonal,
${\cal D}(\rho)={\tilde \rho}:={\rm diag}(\rho)$,
which can be interpreted as a probability vector.
Hence such a \emph{decoherence} process \cite{Sch19},
equivalent to stripping away all off-diagonal elements of the density matrix,
projects the set of quantum states $\Omega_{_N}$
into the classical probability  simplex $\Delta_N$.

An analogous process applied to any
Choi matrix $D_{\Phi}$ produces a diagonal matrix $\tilde D$.
Its diagonal forms a  vector $\vec{t}$ of length $N^2$ which
 can be reshaped into a classical transition matrix
$T(\Phi)=T(D^R)$.
More formally, writing  $\vec{t} = \diag(D_{\Phi})$,
we have $ T_{ij}=\vec{t}_k$,
 where $k=N(i-1)+j$, with $i,j =1,\dots, N$.

The process of stripping of the off-diagonal elements of the Choi matrix
can thus be called  decoherence process
or also hyper-decoherence, to emphasize that it acts
on a larger space of generalized states or quantum maps \cite{Zy08}.
More precisely, we describe it in the formalism of super-maps,
which send the set of all quantum operations into itself, $\Phi'=\Gamma(\Phi)$
-- see \cite{CDAP08,Zy08,Go19}. The hyper-decoherence  ${\c D}_h$
acting on a map $\Phi$  corresponds to a coarse graining supermap
and technically it can be described by the procedure of reshuffling, $D=\Phi^R$, decohering and then reshaping the diagonal matrix $\tilde D$,
 so one can put the above recipe in a nutshell as:
$ T\left(\Phi\right)=\c D_{h}\left(\Phi\right)$.
Alternatively, one may write equations for the discrete evolution of populations,
 $T_{ij} = {\rm Tr}[|i\rangle \langle i| \; {\Phi}(|j\rangle \langle j|)] $.
In this way we represent the process of hyper-decoherence
acting in the space of quantum maps
by a projection from $N^2$ dimensions to $N$.
Let $\Pi_N=\sum_{i=1}^N|i,i\rangle \langle i,i|$
denote a projection operator of dimension $N$.
Then the classical transition matrix $T$ and its entries
can be written by
 \begin{equation}
 \label{supermap}
 T\left(\Phi\right)=  \Pi_N \Phi \Pi_N , \  \ {\rm and \ \ }
 T_{ij} =\Phi_{\substack{ii\\ jj}} ,
 \end{equation}
 where the four index notation introduced in the previous section
is used. Note that no sum  is performed over the repeated indices.

If a Choi matrix $D$ is positive and satisfies the partial trace condition,
Tr$_A D= {\mathbb I}$,  it represents a stochastic map $\Phi$.
Then the corresponding transition matrix $T(\Phi)$ is stochastic,
so that $T_{ij}\ge 0$ and $\sum_i T_{ij}=1$.
$\c S_N^Q$ and $\c S_N^C$ represent the sets of
$N$-dimensional quantum and classical stochastic processes, respectively.
The condition of unitality,  $\Phi({\mathbb I})={\mathbb I}$,
 analogous to the trace preserving property,
 is equivalent to the dual partial trace condition,
Tr$_B D = {\mathbb I}$.
Any completely positive,  trace preserving and unital map
${\Psi}_B$ is called \emph{bistochastic} and  ${\cal B}_N^Q$  will
denote the set of quantum bistochastic maps acting
 on $N$ dimensional system.
Decohering any bistochastic map
one obtains a \emph{bistochastic} transition matrix  $B=T({\Psi}_B)$,
such that $B_{ij}\ge 0$ and $\sum_i B_{ij}=\sum_j B_{ij}=1$
-- see \cite{BZ17}.
Due to the theorem of Birkhoff any bistochastic matrix
can be represented as a convex combination of
permutation matrices.
The set of bistochastic matrices, denoted by  ${\cal B}_N^C$,
has $(N-1)^2$ dimensions and is called the \emph{Birkhoff polytope}.

Consider any unitary matrix $V$ of order $N$.
Unitarity condition, $VV^{\dagger}={\mathbb I}$, implies
that the matrix given by the Hadamard (element-wise) product,
$B=V \odot {\bar V}$, is bistochastic,
as its entries read $B_{ij}=|V_{ij}|^2$.
Thus one may ask, whether any bistochastic  matrix $B$
can be represented in such a way.
The answer is negative, so one introduces the set   ${\cal U}_N^C$,
of \emph{unistochastic matrices}, for which
there exist a unitary $V$ such that $B=V \odot {\bar V}$.

It is easy to check that for $N=2$ any bistochastic matrix
is unistochastic, so both sets coincide, ${\cal U}_2^C ={\cal B}_2^C$.
However, already for $N=3$ there exist bistochastic matrices
which are not unistochastic, so that
${\cal U}_3^C  \subset {\cal B}_3^C$.
As a simple example let us mention the bistochastic matrix considered by Schur,
$B_S=\frac{1}{2}(X_3+X_3^2) \notin {\cal U}_3^C$, where
$X_3$ is the three--element cycle permutation matrix.
In general, to verify whether a given bistochastic matrix
 $B\in {\cal B}_3^C$ is unistochastic one can compute
 the quantity $\cal Q$, related to the Jarlskog invariant  \cite{jarlskog,volume},
 and check if it is positive -- see Appendix \ref{AppA}.
For larger dimensions no explicit criteria necessary and sufficient
for unistochasticity are known \cite{BEKTZ05},
but in the case $N=4$ one can rely on an efficient
numerical procedure based on the algorithm of Haagerup \cite{RGZ18}.

As unitary matrix $V$ of size $N$ determines a bistochastic matrix
$B=V \odot {\bar V} \in {\cal B}_N^C$,
a unitary matrix $U$ of size $N^2$ determines a
certain quantum bistochastic channel.
Physically, such an operation
corresponds to a coupling of the principal system with
an $d$-dimensional ancillary subsystem $E$, initially in the maximally
mixed state, by a non-local unitary matrix $U$,
followed by the partial trace over the environment,
\begin{equation}
\Psi_U = \tr_E \Bigl[U\Bigl(\rho\otimes 
 \frac{\mathbb{I}_d}{d}
 \Bigr)U^{\dagger}\Bigr],
  \label{eq:defuni}
\end{equation}
By construction such maps are trace preserving and unital
for any dimension $d$ of the environment.
The case $d=N$ is distinguished,
as in this case the rank of the corresponding Choi matrix, $r\le d^2$
can be full, so in analogy to the unistochastic matrices
such quantum maps were called \emph{unistochastic channels}
\cite{unistochasticos3}.
The set of unistochastic maps, denoted by
${\cal U}_N^Q$, forms a proper subset of
the set ${\cal B}_N^Q$ of bistochastic maps.
In the one-qubit case, the set of  bistochastic maps
forms the tetrahedron of Pauli channels $\Phi_{\vec{p}}$,
while  ${\cal U}_2^Q$, forms its non-convex subset,
bounded by ruled surfaces \cite{NA07,unistochastic}
and called \emph{Steinhaus tetrahedron} \cite{St99} -- see Fig. \ref{tetra2}
and ref. \cite{PRZ19}.
The set of one-qubit accessible maps was analyzed recently
in \cite{FGL20,JSP20} from a slightly different perspective.

\begin{figure}[ht]
    \centering
    \includegraphics[scale=0.3]{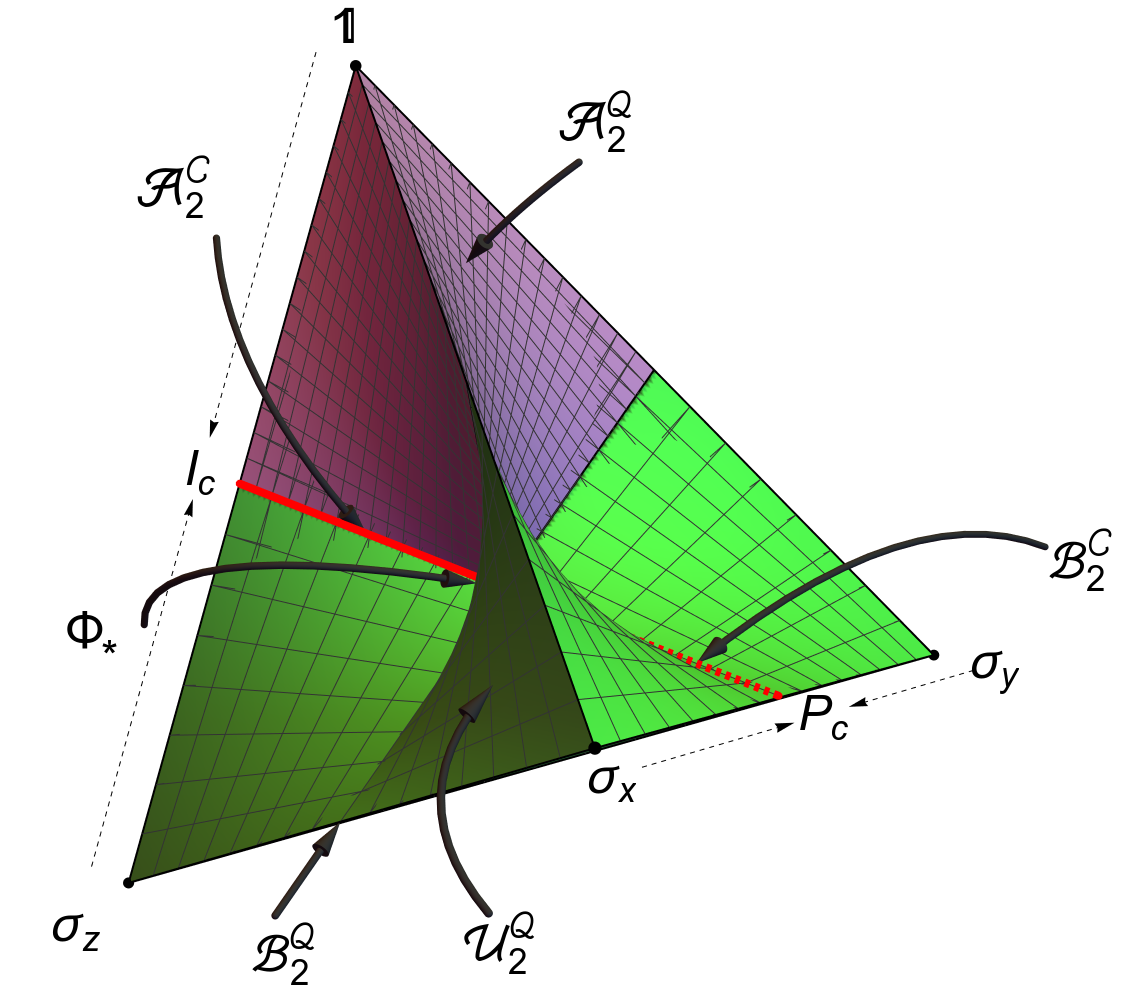}
    \caption{ Any one-qubit mixed unitary channel
     is unitarily equivalent to a convex mixture of Pauli
     rotations, represented  by the tetrahedron ${\cal B}_2^Q$ of Pauli channels.
     The set of  $N=2$ unistochastic channels, denoted by
      ${\cal U}_2^Q$, forms a non-convex set with ruled surfaces,
     called  Steinhaus tetrahedron. Its quarter containing the corner
     representing identity,  forms the set  ${\cal A}_2^Q$
     of  Pauli channels accessible by a quantum semigroup.
     This set is star--shaped with respect to the maximally depolarizing channel $\Phi_*$.
     The set of classical bistochastic matrices  ${\cal B}_2^C$,
     represented by the (red) interval  joining  classical identity $I_c$
     and classical permutation matrix $P_c$,
     includes the set  ${\cal A}_2^C$ of matrices accessible by a classical semigroup.
     This set can be obtained by hyper-decoherence (denoted by dashed arrows),
     which squeezes the 3D set ${\cal A}_2^Q$ into the interval  ${\cal A}_2^C$.
       }
    \label{tetra2}
\end{figure}

The set of unistochastic channels was found
important by studying the maps accessible by a semigroup,
as it was shown \cite{PRZ19}
that  ${\cal A}_2^Q \subset {\cal U}_2^Q$.
 In this work we show  that
 an analogous relation does not hold for $N=3$
 and analyze classical analogues of these sets.
  For any unitary dynamics, $\Psi_V^{\cal I}= V \otimes {\bar V}$,
 the corresponding classical transition matrix is  unistochastic,
  $T=V \odot {\bar V}$. A stronger property  states \cite{KCPZ18}
 that a classical transition matrix $T$ can be coherified to a unitary dynamics
 if and only if  it is unistochastic,  so the set ${\cal I}_N$
 of isometries -  unitary quantum rotations -
 decoheres to the set of classically  unistochastic matrices.

Therefore, although the word \emph{unistochastic}
in both cases refers to the underlying unitary matrix,
the sets of quantum unistochastic  operations ${\cal U}_N^Q$
 and classical unistochastic matrices  ${\cal U}_N^C$
 are not directly related by decoherence.
To explain this fact consider an arbitrary unistochastic
map $\Psi_U$,  determined  in (\ref{eq:defuni})
by a unitary matrix $U$ of order $N^2$.
It is convenient to represent it in  the four-index notation,
$U_{\substack{\alpha\beta\\\gamma\delta}}=
\langle \alpha, \beta |U| \gamma , \delta\rangle$.
It is known \cite{unistochasticos3}
that the corresponding Choi matrix  reads,
  $D_{\Psi_U}=(U^R)^{\dagger}U^R/N$.
Taking its diagonal and reshaping it we arrive at
the classical transition matrix $T(\Psi_U)$ with entries
\begin{equation}
  T_{ij}= \frac{1}{N}\sum_{b,c=1}^{N}
    \vert U_{\substack{ib\\jc}}\vert^2,\quad i,j=1,\dots, N.
    \label{eq:transitionmatrix}
\end{equation}
Note that each entry of $T$ can be written as an average of the entries of a
block of size $N$ of a unistochastic matrix $U\odot {\bar U}$ of order $N^2$.
The above observation provides a clear motivation
 to introduce a family of generalized unistochastic matrices.
 Define an auxiliary $kN \times N$ rectangular matrix $L_{N,k}$,
 being a suitable extension of identity,
\begin{equation}
L_{N,k} =  {\mathbb{I}_N}\otimes |{\bar \phi_k}\rangle ,
    \label{eq:Lkn}
\end{equation}
 where the $k$-dimensional state of the uniform superposition reads,
$|{\bar \phi_k}\rangle=\frac{1}{\sqrt{k}} \sum_{j=1}^k |j\rangle$.
For instance, for
$N=2$ and $k=1$ one has $L_{2,1}= {\mathbb{I}_2}$
while for $k=2$ such a coarse graining matrix
reads $L_{2,2}^T=\frac{1}{\sqrt{2}}\begin{pmatrix}
 1&1&0&0\\
 0&0&1&1
  \end{pmatrix}$.

  \medskip
  {\bf Definition}. A bistochastic matrix $B$ or order $N$
will be called $k$--\emph{unistochastic} if there exists an unitary matrix $U_{kN}$ of size $kN$
such that $B$ is given by the following matrix coarse graining \cite{LGR08}
of the larger unistochastic matrix of order $kN$,
\begin{equation}
  B  = L^T_{N,k} (U_{kN} \odot {\bar U_{kN}}) L_{N,k} .
    \label{k-uni}
\end{equation}

\medskip
Since  $L_{N,1}= {\mathbb{I}_N}$ this
definition implies that any standard unistochastic matrix is $1$-unistochastic,
while expression  (\ref{eq:transitionmatrix})
for the classical transition matrix is equivalent to
$T  = L_{N,N}^T (U_{N^2} \odot {\bar U_{N^2}}) L_{N,N}$.
Thus a classical transition matrix $T$
obtained by decoherence of any unistochastic channel
is $N$--unistochastic.
By construction the set  $ {\cal U}_{N,N}^C$
 of $N$--unistochastic matrices of order $N$
includes the original set
$ {\cal U}_N^C$ of unistochastic matrices.
For $N=2$ every bistochastic matrix  is unistochastic,
but already for $N=3$ there exits  a unistochastic operation
 $\Psi_U$
which decoheres to a non-unistochastic classical transition matrix $T$,
which is $3$--unistochastic -- see Appendix  \ref{AppB}.

\medskip

\begin{table}[h]
\caption{Set  ${\cal S}_N^Q$ of quantum stochastic maps
acting on $N$ dimensional states contains subsets of
bistochastic maps  ${\cal B}_N^Q$,
unistochastic maps $ {\cal U}_N^Q$,
isometric unitary rotations $ {\cal I}_N^Q$,
mixed unitary Weyl channels  $ {\cal W}_N$
defined in (\ref{mixeduni})
and its subset  $ {\cal A}_N^Q$ of maps accessible by a semigroup.
Their classical analogues
forming subsets of the set
${\cal S}_N^C$ of stochastic transition matrices  of order $N$ read:
bistochastic matrices ${\cal B}_N^C$,
$N$-unistochastic matrices $ {\cal U}_{N,N}^C$,
unistochastic matrices $ {\cal U}_N^C$,
circulant bistochastic matrices ${\cal C}_N$ and
 circulant transition matrices $ {\cal A}_N^C$ accessible by a
 classical  semigroup.
Vertical arrows represent  process of hyper-decoherence, e.g.
${\cal D}_h({\cal S}_N^Q)={\cal S}_N^C$ and
${\cal D}_h(\Phi_p)=T_{q}$
with $q$ being the marginal (\ref{q}) of the vector $p$,
the rectangular coarse graining matrix $L_{N,N^2}$ of order $N \times N^2$
is defined in   (\ref{eq:Lkn}),
while $U$ and $V$ represent unitary matrices of size $N^2$ and $N$,
respectively.
}
  \smallskip
\hskip -0.2cm
{\renewcommand{\arraystretch}{1.47}
\begin{tabular}
[c]{| l | c  c  c  c  c  c  c   | c   | c |} 
\hline \hline
   \em       Quantum  operation --     &   ${\cal S}_N^Q$    &  {\large $\supset$} &   ${\cal B}_N^Q$   &  {\large $\supset$} &  ${\cal U}_N^Q$
                                             &  {\large $\supset$} &  ${\cal I}_N^Q$   &
                                          ${\cal W}_N$    &   ${\cal A}_N^Q$
    \\
  \em  stochastic map,  $\Phi^R=D \ge 0$  ~     &    ~ ${\rm Tr}_A D= {\mathbb I}$ ~  & &   ~${\rm Tr}_B D= {\mathbb I}$ ~ & &
        ~ $D=N^{-1} (U^R)^{\dagger} U^R$ ~  & &  ~
                                             $D=(V\otimes {\bar V})^R$ ~
                                              & ~$\Phi_p$ ~&
                                             ~ $\Phi=\exp(t {\cal L}) $ ~
    \\
    \hline
    \em ~~\emph{hyper-decoherence} ${\cal D}_h (\; .\; )$     &  $\downarrow$   &  & $\downarrow$ & &  $\downarrow$ & & $\downarrow$ &  $\downarrow$  & $\downarrow$  
     \\
\hline
   \em      Classical transition       &   ${\cal S}_N^C$   &   {\large $\supset$} &   ${\cal B}_N^C$   &  {\large $\supset$} &
     $ {\cal U}_{N,N}^C $ 
                                                   &   {\large $\supset$ }
           & $ {\cal U}_{N}^C$
            & ${\cal C}_N $  &   ${\cal A}_N^C$ 
    \\
       \em stochastic  matrix, ~$T_{ij}\ge 0$         &   ~ $\sum_i T_{ij}=1 $ ~  & &   ~$\sum_j T_{ij}=1 $ ~ &  &
        ~ $T=L^T ( U\odot {\bar U})L $ ~   & &  ~
                                             $T= V\odot {\bar V}$    &  $T_q$ &   $T=\exp(t {\cal K}) $
                                              \\
\hline \hline
\end{tabular}
}
\label{tab:sets}
\end{table}

\medskip

A list of sets  of quantum operations discussed further in this work
and the corresponding sets of classical transition matrices
is presented in Table \ref{tab:sets}.
Note that the inclusion relations for quantum maps,
visible in the table,
correspond to analogous relations in the classical case.
Natural decoherence relations, like
${\cal D}_h ({\cal S}_N^Q)={\cal S}_N^C$
and
${\cal D}_h ({\cal B}_N^Q)={\cal B}_N^C$,
connect the sets of unistochastic channels and
$N$-unistochastic matrices,
${\cal D}_h ({\cal U}_N^Q)={\cal U}_{N,N}^C$.
The latter one contains the set
${\cal U}_{N}^C$ of unistochastic matrices arising
  from isometric unitary quantum maps  by hyper-decoherence,
${\cal D}_h ({\cal I}_N)={\cal U}_N^C$,
which can be modeled by Schur superchannels \cite{Pu20}.
The set of mixed unitary Weyl channels
decoheres to the set of circulant bistochastic matrices,
${\cal D}_h ({\cal W}_N)={\cal C}_N$,
while quantum accessible Weyl channels
decohere to classical accessible bistochastic
matrices,
${\cal D}_h ({\cal A}_N^Q)={\cal A}_N^C$.

Any unitary rotation forms a unistochastic map, so that
$ {\cal U}_N^Q \supset {\cal I}_N^Q$.
A classical analogue of unitary rotations is played by permutations matrices
and their set sits inside the set of unistochastic matrices,
$ {\cal P}_N^C  \subset {\cal U}_N^C \subset {\cal B}_N^C$.
In the case $N=2$ the following relations hold:
$ {\cal A}_2^Q \subset {\cal U}_2^Q$
and $ {\cal U}_2^C  = {\cal B}_2^C$, which are not true for higher $N$.
The set ${\cal U}_N^Q$ of unistochastic channels and the set
${\cal W}_N$ of the Weyl channels are incomparable.
So are the sets of Weyl channel and accessible channels,
so in this work we will restrict our attention to the Weyl channels
which are accessible by a semigroup.

\section{General case: accessibility by a semigroup for an arbitrary dimension $N$}
\label{sec:two}

  \subsection{Channel eigenvalues and probabilities}
In order to answer the question, whether a given map is
accessible by a semigroup we need to analyze further properties of
channels in the form  \eqref{phi}.
The following lemma presents  a key
feature of mixed unitary channels written in the Weyl basis.
\begin{lemma}
\label{commutativity}
Weyl unitary operators satisfy  $[U_{ij}\otimes\overline{U}_{ij},U_{kl}\otimes\overline{U}_{kl}]=
U_{ij}U_{kl}\otimes\overline{U}_{ij}\overline{U}_{kl}-
U_{kl}U_{ij}\otimes\overline{U}_{kl}\overline{U}_{ij}=0$
for any choice of the indices $i,j,k,l=0,\dots, N-1$.
 The equality holds as Weyl unitaries are commutative up to a phase,
  which is not relevant here as the expressions above contain
   $U$ and its complex conjugate, $\overline{U}$.
\end{lemma}
 Accordingly, all  combinations  (not necessarily convex) of
 tensor products of Weyl unitaries and their complex conjugates
 are compatible, so are all mixed unitary channels
 $\Phi_{\vec{p}}$, i.e. $\forall\vec{p},\vec{p}'\ \
 [\Phi_{\vec{p}},\Phi_{\vec{p}'}]=0$. For the sake of
 brevity, from now on we may drop subscript $\vec{p}$ and denote
 these channels by $\Phi$.  The common eigenbasis of the
 channels defined in Eq. \eqref{phi} is given by the maximally
 entangled states proportional to the vectorized form of Weyl
 unitaries.
 For any matrix $A=\sum A_{ij}\outprod ij$ let us use
 the shorthand  notation
 $\Ket A=\sum A_{ij}\ket i\ket j$ to
 represent its vectorized form. It  enjoys
 the following properties: $(A\otimes C^T)\Ket B=\Ket{ABC}$
 and $\Bra B A\rangle\rangle=\tr{(B^\dagger A)}$,
 which  defines the Hilbert-Schmidt inner product in the space of matrices.
A matrix element of the superoperator $\Phi_{\vec{p}}$
in the non-normalized basis $\{\Ket{U_{ij}}\}$ reads,
\begin{eqnarray}\label{eigphi}
\nonumber\Bra{U_{ij}}\Phi_{\vec{p}}\Ket{U_{mn}}&=&\sum_{kl} p_{kl}\Bra{U_{ij}}U_{kl}\otimes \overline{U}_{kl}\Ket{U_{mn}}=\sum_{kl} p_{kl}\tr\big(U_{ij}^\dagger U_{kl}U_{mn}U_{kl}^\dagger\big)\\&=&\sum_{kl}
 p_{kl} \omega_{_N}^{ml-kn} 
 \tr\big(U_{ij}^\dagger U_{mn}\big)=N\sum_{kl} p_{kl}
 \omega_{_N}^{ml-kn}     
 \delta_{im}\delta_{jn}.
\end{eqnarray}
The above equalities holds due to
orthogonality of Weyl unitaries
 and their commutativity up to a phase mentioned above.
For any mixed unitary channel \eqref{mixeduni}
represented by a superoperator in the form
$\Psi_{\vec{p}}=\sum p_{n}V_{n}\otimes \overline{V}_{n} $,
the dynamical matrix $D$, obtained by reshuffling its entries,
\begin{eqnarray}
D_{\Psi}= \Psi_{\vec{p}}^R=\sum_{n}p_{n}\big(V_{n}\otimes \overline{V}_{n}\big)^R=\sum_{n}p_{n} \Ket{V_n}\Bra{V_n},
\end{eqnarray}
 corresponds to the map by the  Choi-Jamio\l kowski isomorphism \cite{BZ17}.
Note that the last form yields  the spectral decomposition of the Choi matrix
 $D_{\Psi}$ if and only if the set
$\{V_n\colon V_n\in U(N)\}_{n=1}^{N}$ forms an orthogonal basis.
Observing that $\langle\langle U_{kl}|U_{ij}\rangle\rangle=N\delta_{ik}\delta_{jl}$
and  $\Phi_{\vec{p}}^R=\sum_{kl}p_{kl}\Ket{U_{kl}}\Bra{U_{kl}}$ we infer
 that all the channels $\Phi_{\vec{p}}$ are compatible
 and also that they commute with the set of the corresponding dynamical maps.
This fact implies  the following proposition clarifying the relation between
the eigenvalues  $\lambda_{kl}$ of the superoperator $\Phi_{\vec{p}}$
and the probabilities, $p_{kl}$,  which  are proportional to the
eigenvalues of the corresponding dynamical matrix $D=\Phi_{\vec{p}}^R$.

\begin{proposition}\label{cdphi}
There exists a complex Hadamard  matrix $H$
(unitary up to rescaling  with all unimodular entries),
 which  transforms probabilities $p_{kl}$  into eigenvalues  $\lambda_{mn}$
 of the superoperator $\Phi_{\vec p}$ defined in Eq. \eqref{phi}, namely
   $\vec{\lambda}=H\vec{p}$ and  $\vec{p}=\frac{1}{N^2}H\vec{\lambda}$.
   The matrix $H$ of order $N^2$ is hermitian and its entries read \cite{CW13,CW15},
\begin{eqnarray}
\label{M}
H_{\stackrel{\scriptstyle mn}{kl}}=\omega_{_N}^{\; ml-kn},
 \ \ k,l,m,n=0,\dots, N-1.
\end{eqnarray}
\end{proposition}
To demonstrate the above statement one needs to show that
$H^2=HH^\dagger=N^2 \Id$, which means that $H$
is a complex Hadamard matrix \cite{TZ06},
self inverse up to the constant $N^2$.
Hermiticity of $H$ is obvious from definition,
and for unitarity we have
\begin{eqnarray*}
\sum_{k,l=0}^{N-1}H_{\stackrel{\scriptstyle mn}{kl}}
H_{\stackrel{\scriptstyle kl}{pq}}=\sum_{kl}\; \omega_{_N}^{ml-kn}\omega_{_N}^{kq-lp}=\sum_{l}\omega_{_N}^{(m-p)l}\; \sum_k \omega_{_N}^{(q-n)k}=N^2\delta_{mp}\delta_{nq}.
\end{eqnarray*}

\begin{remark}\label{remark}
For any  quantum channel of the form $\Phi_{\vec{p}}=\sum p_{kl}U_{kl}\otimes\overline{U}_{kl}=\frac{1}{N}\sum\lambda_{kl}\Ket{U_{kl}}\Bra{U_{kl}}$, the following facts can be verified
with the help of the matrix $H$,
\begin{enumerate}
\item For any superoperator $\Phi_{\vec{p}}$ of this form its leading eigenvalue (corresponding to invariant state $\frac{1}{\sqrt{N}}\Ket\Id$) is equal to unity, $\lambda_{00}=1$.
\item For completely depolarizing channel $\Phi_\ast$ corresponding to
the flat probability vector, $p_{kl}=\frac{1}{N^2}$ for $k,l=0,\dots, N-1$,
 all subleading eigenvalues vanish, $\lambda_{kl}=0$ for $k\neq0$ and $l\neq0$.
\item $\lambda_{k,l}=\overline{\lambda}_{-k\oplus N,-l\oplus N}$,
so apart from  $\lambda_{00}$, for an even $N$ there exist three other eigenvalues,
$\lambda_{\frac{N}{2},0}$, $\lambda_{0,\frac{N}{2}}$, and
$\lambda_{\frac{N}{2},\frac{N}{2}}$ which are always real.
\item As the entries of $H$ are unimodular and the eigenvalues of the superoperator
$\Phi_{\vec p}$ are their convex combinations, they belong to the unit disk,
 $|\lambda_{k,l}|\leq1$ for any $k,l$.
\end{enumerate}
\end{remark}

  \subsection{Lindblad generators}
Knowing all we need about the spectral form of quantum
channels and their dynamical maps, consider now a special
\emph{family} of dynamical semigroups
$\Lambda_{t_\mu}=\e^{t_\mu\mathcal{L}_\mu}$, defined by the
non-negative interaction time  $t_\mu$ and the Lindblad
generator $\mathcal{L}_\mu$ specified by a single jump
operator $L_i=U_\mu\delta_{i\mu}$.
Due to Eq.~\eqref{lindblad2},
the generator $\c{L}_{\mu}$ is
\begin{equation}\label{lmu}
\mathcal{L}_\mu=U_\mu\otimes\overline{U}_\mu-\Id_{N^2}.
\end{equation}
Note that for the case $\mu=0$ Lindblad generator is trivial,
$\c L_0=0$.
By the definition of $\mathcal{L}_{\mu}$
and Lemma~\ref{commutativity},
the Lindblad generators do commute,
$[\mathcal{L}_\mu ,\mathcal{L}_\nu]=0$.
The commutativity of Lindblad
generators implies commutativity on their respective
dynamical semigroups,
$\Lambda_{t_\mu}\Lambda_{t_\nu}=
\Lambda_{t_\nu}\Lambda_{t_\mu}$.
So it is convenient to write:
\begin{eqnarray}
\label{lt}
\Lambda_{t_0}\Lambda_{t_1}\dots\Lambda_{t_{N^2-1}}&=&
\e^{t_0\c L_0} \e^{t_1\mathcal{L}_1}\dots\e^{t_{N^2-1}\mathcal{L}_{N^2-1}}=
\e^{\sum_\mu t_{\mu}\mathcal{L}_{\mu}}=\e^
{t\mathcal{L}}=\Lambda_{t}.
\end{eqnarray}
Here $t=\sum_{\mu} t_\mu$ and
\begin{equation}\label{p'}
\mathcal{L}=\sum_{\mu} p'_\mu\mathcal{L}_\mu,\ \ \rm{where:} \ \ \ p'_\mu=\frac{t_\mu}{t}.
\end{equation}
Regarding that $\c L_0=0$, $t_0$ has no physical meaning and
one can set it to zero, $t_0=0$, which means the first component of the
$N^2$ dimensional probability vector $\overrightarrow{p'}$
 is zero, $p'_0=0$.
Above representation of the semigroup $\Lambda_t$ implies
the following statement.
\begin{corollary}
\label{logconv}
For any dimension $N$  the set ${\cal A}_N^Q$
of quantum channels accessible by a  dynamical semigroup is
\emph{log convex}: any generator $\c{L}$ in the exponent can be represented by  a
convex combination of $N^2-1$ generators $\c{L}_{\mu}$.
\end{corollary}

Moreover, it is known that
any concatenation of quantum channels is a quantum channel itself.
However, for the family of semigroups in Eq.~\eqref{lt}, the stronger fact is
that quantum channels gained by an arbitrary concatenation
of dynamical semigroups, $\e^{t_\mu\mathcal{L}_\mu}$, are mixed unitary channels in
the sense of Eq. \eqref{phi}, i.e.
$\Lambda_{t}=\Phi$,  and the order
of dynamical semigroups is not important due to
commutativity. To prove the claim,
by applying Eq. \eqref{lmu}, we obtain $\mathcal{L}=\sum_\mu
p'_\mu U_\mu\otimes\overline{U}_\mu-\Id_{N^2}$,
so the spectral form of the Lindblad generator is:
\begin{eqnarray}\label{lindbladspec}
\mathcal{L}=\frac{1}{N}\sum_{\nu=0}^{N^2-1}(\lambda_\nu^{\prime}-1)\Ket{U_\nu}\Bra{U_\nu},
\end{eqnarray}
where $\lambda_\nu^{\prime}$
 are the eigenvalues of the operator $\sum p'_\mu U_\mu\otimes\overline{U}_\mu$
obtained  by acting with the matrix $H$ defined in Eq. \eqref{M} on the
 probability vector  $p'_\mu$ of rescaled interaction times.
In conclusion, the dynamical semigroup associated with $\mathcal{L}$ has the form:
\begin{equation}\label{ltexplicit}
\Lambda_{t}=\frac{1}{N}\sum_{\nu=0}^{N^2-1}\e^{t(\lambda_\nu^{\prime}-1)}
\Ket{U_\nu}\Bra{U_\nu}.
\end{equation}

Now we will show that the semigroup ~\eqref{ltexplicit} can be written in the form \eqref{phi}
as a convex combination of Weyl unitary matrices.
   The probabilities $p_{\mu}$ in the latter equation
   are obtained by applying $H$ to the eigenvalues of
   $\Lambda_{t}$, and they are proportional to the eigenvalues of
 $\Lambda_{t}^R$, the reshuffled form of  $\Lambda_{t}$.
 This implies that every entry of the probability vector is non negative, as required.
Accordingly, quantum
channels accessible by a dynamical semigroup $\Lambda_{t}$
form the subset $\mathcal{A}_N^Q$ within the simplex $\Delta_{N^2-1}$
of mixed unitary channels in the Weyl basis.
\par
  Note that if $\lambda_{\nu}'$ is real, then  $\e^{t(\lambda_\nu^{\prime}-1)}$, the
  corresponding eigenvalue of $\Lambda_t$,
  is  positive. Also, if $\lambda_{\nu}'$ and
  $\overline{\lambda'}_{\nu}$  form a complex conjugate pair,
   the corresponding eigenvalues of $\Lambda_t$
   are again a  complex conjugate pair.
In this case, at
$t_+=\frac{2n\pi}{|\rm{Im}\left(\lambda'_\nu\right)|}$
and $t_-=\frac{(2n-1)\pi }{|\rm{Im}\left(\lambda'_\nu\right)|}$ for any
natural number $n\in\mathbb{N}$, the equality
$\e^{t(\lambda_\nu^{\prime}-1)}=
\e^{t(\overline{\lambda'}_\nu-1)}$ holds. So
 $\Lambda_t$ becomes a map with a positive degenerate or a negative
 degenerate eigenvalue, respectively, at $t_+$ and $t_-$ .
Finally,
since $|\lambda_\nu^\prime|\leq1$ (see Remark \ref{remark}),
the moduli of the eigenvalues of $\Lambda_t$ decay
exponentially with respect to time $t$ or
remain unchanged and equal to unity.

\begin{theorem}
  Consider a mixed unitary channel represented in the Weyl basis,
    written in the vectorized form $|U_{\mu}\rangle \rangle$,
  \[
    \Phi = \frac{1}{N}\sum_{\mu=0}^{N^2-1} \lambda_{\mu}
    \Ket{U_\mu}\Bra{U_\mu}.
  \]
  In general the eigenvalues $\lambda_{\mu}$ of the superoperator are complex,
  and a
  suitable choice of   $\log\lambda_\mu$ is discussed below.
   The necessary and sufficient
condition for $\Phi$ to be accessible by a dynamical
semigroup \eqref{lt} is that
the sum on the right hand side of Eq. (\ref{time}) is real and positive,
\begin{eqnarray}
\label{time}
  t_\mu=\frac{1}{N^2}\sum_{\nu=1}^{N^2-1}\big(\pr
    H\pr)_{\mu\nu}\log\lambda_\nu\geq0 , 
\end{eqnarray}
for  $\mu=1,2,\ldots,N^2-1$.
Here $H$ denotes the complex Hadamard matrix defined in \eqref{M},
while  $\pr=\Id_{N^2}-\project{00}$ is a projector operator onto $(N^2-1)$-dimensional space
and $\project{00}$ written as a product basis projects $\vec{\lambda}$ to its leading element $\lambda_0=1$.
\end{theorem}

Before proceeding to the proof note that if there exists a single real
negative eigenvalue of $\Phi$, its logarithm is complex, so the above
condition  cannot be satisfied, hence the map is not accessible.
On the other hand, if negative eigenvalue is degenerated
(an even number of times) then the imaginary parts of log$\lambda_{\nu}$
can cancel and the sum  (\ref{time}) can be real and positive.
The fact that a one-qubit map with a negative degenerated
eigenvalue of the superoperator can be accessible by a semigroup
was already reported in  \cite{DZP18}.

\begin{proof}
To prove the \enquote{if} part, we will assume $\Phi$ is accessible by a
dynamical semigroup~\eqref{lt}, $\Phi=\Lambda_t$ for some interaction
time $t\geq0$, and apply Eq. \eqref{ltexplicit}:
\begin{eqnarray}\label{mprime}
\log\lambda_\mu=(\lambda_\mu^\prime-1)t=
(\sum_\nu H_{\mu\nu}p'_\nu-1)t=
\sum_\nu\left(H_{\mu\nu}-1\right)t_\nu=
\sum_\nu H_{\mu\nu}^\prime t_\nu,
\end{eqnarray}
where $H'=H-W$ with $W_{\mu\nu}=1$. Recalling the definition
of matrix $H$, see Eq.~\eqref{M}, it is clear that
the entries of its first row and column are equal to $1$, so
the corresponding
entries of $H^\prime$ vanish. Consequently,  $H^\prime$ has
the same structure of a projector onto a
$(N^2-1)$-dimensional subspace,
in other words, $H^\prime=\pr
H^\prime\pr$. This means that $H'$ is not invertible,
however, one
can define its inverse on the $(N^2-1)$-dimensional subspace  $\Pi$.
We proceed to show that this inverse exists and is equal to $\pr H\pr$, up
to the constant $N^2$,
\begin{eqnarray}
\nonumber(\Pi H\Pi)(\Pi H^\prime\Pi)&=&\sum_{\stackrel{\scriptstyle m n k l}{ m'n'k'l'}}
\omega_{_N}^{ml-nk}  
\left(\omega_{_N}^{m'l'-n'k'} 
-1\right)\delta_{km'}\delta_{ln'}\outprod{mn}{k'l'}\\ \nonumber&-&
\sum_{\stackrel{\scriptstyle m n }{ m'n'k'l'}}
\left(\omega_{_N}^{m'l'-n'k'} 
-1\right)\left(\delta_{m'0}\delta_{n'0}\outprod{mn}{k'l'}+\delta_{mm'}\delta_{nn'}\outprod{00}{k'l'}\right)\\ &+&\sum_{\stackrel{\scriptstyle m' n' }{ k'l'}}
\left(\omega_{_N}^{m'l'-n'k'}
-1\right)\delta_{m'0}\delta_{n'0}\outprod{00}{k'l'}=N^2\left(\sum_{mn}\project{mn}-\project{00}\right).
\end{eqnarray}
By multiplying both sides of  Eq. \eqref{mprime}
by $\Pi H\Pi$,
one gets the equality in Eq.~\eqref{time}. Because
$t_\mu\geq 0$ for $\mu=1,\dots,N^2-1$, we completely recover
Eq~\eqref{time}. Note that Eq. \eqref{mprime} implies
$\log\overline{\lambda}_\mu=\overline{\log\lambda_\mu}$.

To prove \enquote{then} part, we should  show that if inequalities \eqref{time}
 between eigenvalues hold, the channel is accessible by  dynamical
semigroups~\eqref{lt}. This fact is actually obvious as we have explicitly
 introduced the interaction time and this completes the proof.
\end{proof}

Note that the above result allows us not only to
check, whether a given channel is accessible,
but if the answer is positive, one can also construct the desired semigroup.
For any probability vector  ${\vec p}$,
determining the mixed unitary channel  (\ref{mixeduni}),
we find the complex vector, $\vec{\lambda}=H\vec{p}$,
check if conditions (\ref{time}) are satisfied,
and if it is the case, we use the non-negative vector $\vec t$
of interaction times to write the semigroup (\ref{lt}).

Moreover, as it should be, Eq.~\eqref{time} is equivalent of existence
of a proper generator for the quantum channel $\Phi$
 that satisfies properties mentioned in the section \ref{lindbladchar}.
For accessible channels from $\c A^Q_N$, the (non-real) logarithm of $\Phi$
providing such a generator is equal to:
\begin{eqnarray}\label{logphi}
\log\Phi=\frac{1}{N}\sum_\mu\log{\lambda_\mu}\Ket{U_\mu}\Bra{U_\mu}.
\end{eqnarray}
The above equation fulfils trace preserving condition and by considering
$\log\overline{\lambda}_\mu=\overline{\log\lambda_\mu}$,
 Hermiticity preserving property  also holds.
Now, we  need to demonstrate that  $\log\Phi$ is
conditionally completely positive.
Note that
$\left(\log\Phi\right)^R=\sum_\mu
C_\mu\Ket{U_\mu}\Bra{U_\mu}$, where $C_\mu=\frac{1}{N^2}\sum_\nu
H_{\mu\nu}\log{\lambda_\nu}$, and that
the maximally entangled state
$\ket{\psi_+}$ is equal to $\frac{1}{\sqrt{N}}\Ket{U_0}$ where
$U_0=\Id_N$. So conditionally completely positivity for
$\log\Phi$ defined by Eq. \eqref{logphi}  means that all of
 eigenvalues of its reshuffled form are positive
besides the one associated with $\Ket{U_0}$, i.e.
$C_\mu\geq0$ for  $\mu=1,\dots,N^2-1$.
Observing that $\log\lambda_0 =0$, we get
 $\sum_\nu (\Pi H\Pi)_{\mu\nu}\log\lambda_\nu\geq0$ which is satisfied
 as we assumed $\Phi\in\c A^Q_N$.

Thus far,  we have shown which generator for channels in $\c A_N^Q$
is suitable, however, for logarithm of $\lambda_\nu$ one has
 $\log\lambda_\nu=\log
r_\nu+i\theta_\nu+i2\pi M_\nu$ with
$M_\nu\in\mathbb{Z}$ and $\theta_\nu\in[-\pi,\pi]$.
Note that both $\pi$ and $-\pi$ are included in the interval,
so for a negative degenerate eigenvalue one can assign $\pi$
to one of them and $-\pi$ to another one.
 Recalling the condition imposed on the logarithm of
  eigenvalues, one can  verify (with two indices)
 $M_{-k\oplus N,-l\oplus N}=-M_{k,l}$. This implies:
\begin{equation}\label{eneq1}
 t_{m,n}=\frac{1}{N^2}\sum_{k,l}H_{\stackrel{\scriptstyle
 mn}{kl}}\left(\log{r_{kl}}+i\theta_{kl}+i2\pi M_{k,l}\right)\geq 0,
 \end{equation}
 for $m,n\in\{0,\cdots,N-1\}$ which are not simultaneously zero,
 and so for $m'=-m\oplus N$ and $n'=-n\oplus N$:
\begin{equation}\label{eneq2}
 t_{-m\oplus N,-n\oplus N}=
 \frac{1}{N^2}\sum_{k,l}H_{\stackrel{\scriptstyle
 mn}{kl}}\left(\log{r_{kl}}-i\theta_{kl}
 -i2\pi M_{k,l}\right)\geq 0.
\end{equation}
 These two relations result in the following bounds:
 \begin{equation}
 -\sum_{k,l}H_{\stackrel{\scriptstyle
 mn}{kl}}\left(\log{r_{kl}}+i\theta_{kl}\right)\leq
 i2\pi\sum_{k,l}H_{\stackrel{\scriptstyle
 mn}{kl}} M_{k,l}\leq\sum_{k,l}H_{\stackrel{\scriptstyle
 mn}{kl}}\left(\log{r_{kl}}-i\theta_{kl}\right).
 \end{equation}

In the next Remark we summarize further properties of the accessible
channels fulfilling relation \eqref{time} and properties of $M_{k,l}$.

\begin{remark}
Regarding this fact that $\lambda'_\nu$s themselves
are eigenvalues  of a quantum Weyl
channel--see Eq.~\eqref{lindbladspec},
so they have the properties mentioned in Remark \ref{remark},
one can achieve the following results for accessible maps,
 $\Phi=\Lambda_t$:
\begin{enumerate}
\item As it is mentioned earlier in this section for a real $\lambda_{\nu}'$
dynamical semigroup
$\Lambda_t$ has a positive eigenvalue. Due to Remark \ref{remark}
we know for an even $N$  there are three real subleading eigenvalues
$\lambda'_\nu$.
 Consequently, for an even $N$ the essentially real subleading eigenvalues of
 any accessible map $\Phi\in\c A_N^Q$   are necessarily positive,
  i.e. $0\leq\lambda_{mn}\leq1$ for
 $\{(m,n)\}=\{(0,\frac{N}{2}),(\frac{N}{2},0),(\frac{N}{2},\frac{N}{2})\}$.

\item Adding Eq. \eqref{eneq1} to Eq. \eqref{eneq2}, we get the following
necessary condition for eigenvalues of an accessible map:
 \begin{equation}\label{ncon}
 \sum_{k,l}H_{\stackrel{\scriptstyle
 mn}{kl}}\log{r_{kl}}\geq0,
 \end{equation}
 where $r_{kl}$ is the modulus of $\lambda_{kl}$, i.e.
 $\lambda_{kl}=r_{kl}\e^{i\theta_{kl}}$.

 \item If $\forall m,n:\ \ \sum_{k,l}H_{\stackrel{\scriptstyle
 mn}{kl}}\left(\log{r_{kl}}+i\theta_{kl}\right)\geq0$, we do not need to
  check any $M_{k,l}$ to see if the channel is accessible.
  However, if there exist $M_{k,l}$ for which inequality
  \eqref{eneq1} holds for any $m,n$, then the associated quantum channel
  is accessible not by a unique Lindblad generator.

    \item If $\sum_{k,l}H_{\stackrel{\scriptstyle
 mn}{kl}}\left(\log{r_{kl}}+i\theta_{kl}\right)\leq0$ for
 both $(m,n)=(m',n')$ and $(m,n)=(-m'\oplus N,-n'\oplus N)$ the
  associated channel is not accessible.

  \item Applying Eq.~\eqref{mprime} for an even $N$,
   we can show $M_{k,l}=0$ for
  $\{(k,l)\}=
  \{(0,\frac{N}{2}),(\frac{N}{2},0),(\frac{N}{2},\frac{N}{2})\}$.

\end{enumerate}
\end{remark}

Eq.~\eqref{mprime} shows the eigenvalues of a quantum
channel belonging to a dynamical semigroup as an explicit
function of time. By applying matrix $H$, we can transform
eigenvalues into probabilities and gain the following
relation representing the probabilities as an explicit
function of time.
\begin{equation}\label{ptime}
p_\mu=\frac{1}{N^2}\sum_\nu H_{\mu\nu}\e^{\sum_\eta H'_{\nu\eta} t_\eta}.
\end{equation}

Rewriting the previous equation using two indices $kl$ rather than
$\mu$, one has
$p_{kl}=\frac{1}{N^2}\sum_{mn}H_{\stackrel{\scriptstyle kl
}{ mn}}\e^{\sum_{rs}H'_{\stackrel{\scriptstyle mn }{
rs}}t_{rs}}$.
Computing the partial derivative with respect to $t_{ij}$,
we obtain the following useful relations
\begin{eqnarray}\label{prob3}
\nonumber\frac{\partial p_{kl}}{\partial t_{ij}}&=&
\frac{1}{N^2}\sum_{mn}H_{\stackrel{\scriptstyle kl }{ mn}}
H'_{\stackrel{\scriptstyle mn }{ ij}}
\e^{\sum_{rs}H'_{\stackrel{\scriptstyle mn }{ rs}}t_{rs}}=
\frac{1}{N^2}\sum_{mn}H_{\stackrel{\scriptstyle kl }{ mn}}
\left(H_{\stackrel{\scriptstyle mn }{ ij}}-1\right)
\e^{\sum_{rs}H'_{\stackrel{\scriptstyle mn }{ rs}}t_{rs}}\\
 \nonumber &=&\frac{1}{N^2}\sum_{mn}
 H_{\stackrel{\scriptstyle kl }{ mn}}
 H_{\stackrel{\scriptstyle mn }{ ij}}
 \e^{\sum_{rs}H'_{\stackrel{\scriptstyle mn }{ rs}}t_{rs}}-
 \frac{1}{N^2}\sum_{mn}
 H_{\stackrel{\scriptstyle kl }{ mn}}
 \e^{\sum_{rs}H'_{\stackrel{\scriptstyle mn }{ rs}}t_{rs}}\\
 \nonumber
           &=&\frac{1}{N^2}\sum_{mn}H_{\stackrel{\scriptstyle
           kl }{ mn}}H_{\stackrel{\scriptstyle mn }{
       ij}}\e^{\sum_{rs}H'_{\stackrel{\scriptstyle mn }{
   rs}}t_{rs}}-p_{kl}.
\end{eqnarray}
In the next step we get
\begin{eqnarray}\label{prob2}
 \nonumber p_{kl}+\frac{\partial p_{kl}}{\partial t_{ij}}&=&
 \frac{1}{N^2}\sum_{mn}\omega_{_N}^{(kn-ml)} \omega_{_N}^{(mj-ni)}
 \e^{\sum_{rs}H'_{\stackrel{\scriptstyle mn }{ rs}}t_{rs}}\\
 &=&  \nonumber
 \frac{1}{N^2}\sum_{mn}\omega_{_N}^{n(k-i)-m(l-j)}
 \e^{\sum_{rs}H'_{\stackrel{\scriptstyle mn }{ rs}}t_{rs}}
 = \sum_{mn}H_{\stackrel{\scriptstyle -i\oplus k,
 -j\oplus l}{m\ ,\ n}}
 \e^{\sum_{rs}H'_{\stackrel{\scriptstyle mn }{
 rs}}t_{rs}},
\end{eqnarray}
which implies the relation
\begin{equation}\label{prob}
 p_{kl}+\frac{\partial p_{kl}}{\partial t_{ij}}=p_{-i\oplus k, -j\oplus l}.
\end{equation}
This equality  shows that it is enough to know only one
component of the probability vector, say $p_{00}$, to
compute the rest of the probability vector,
$p_{00}+\frac{\partial p_{00}}{\partial
t_{ij}}=p_{-i\oplus N, -j\oplus N}$.
By a proper substitution, this can be applied to
eigenvalues.

 To demonstrate this approach in action consider first the case $N=2$,
 in which the superoperator is expressed by the  Pauli matrices,
  $\Phi=\sum_\mu
 p_\mu\sigma_\mu\otimes\overline{\sigma}_\mu$. In this case,
  the Hadamard matrix $H$ of order $N^2$ defined in \eqref{phi} reads,
 \begin{equation}
 H=\begin{pmatrix}
 1&1&1&1\\
 1&1&-1&-1\\
 1&-1&1&-1\\
 1&-1&-1&1
  \end{pmatrix}.
 \end{equation}
 By acting with $H$ on the probability vector $\vec{p}=(p_0,p_1,p_2,p_3)^T$, one
 obtains eigenvalues as a function of the probabilities \cite{PRZ19}.
 Furthermore, by applying Eq. \eqref{time} in qubit case, one obtains relation
\begin{equation}
 \begin{pmatrix}
 t_1\\
 t_2\\
 t_3
  \end{pmatrix}=\frac{1}{4}\begin{pmatrix}
 1&-1&-1\\
 -1&1&-1\\
 -1&-1&1
  \end{pmatrix}\begin{pmatrix}

 \log{\lambda_1}\\
\log{\lambda_2} \\
\log{\lambda_3}
  \end{pmatrix}=\frac{1}{4}\begin{pmatrix}

 \log{\frac{\lambda_1}{\lambda_2\lambda_3}}\\
\log{\frac{\lambda_2}{\lambda_1\lambda_3}} \\
\log{\frac{\lambda_3}{\lambda_1\lambda_2}}
\end{pmatrix},
  \end{equation}
 in which  only the `core' of size $N^2-1$  \cite{TZ06} of the complex Hadamard matrix $H$ appears.
Non-negativity of all components of the vector $\vec t
=(t_1,t_2,t_3)$ implies the desired set of three inequalities
 \begin{equation}
 \label{lambda3}
 \lambda_i\geq\lambda_j\lambda_k,
 \ \ {\rm for} \ \
i\neq j\neq k.
 \end{equation}
  These relations
   form the necessary and
    sufficient condition for a quantum channel to be a seed
    for Pauli semigroups \cite{DZP18,PRZ19}. By
    applying Eq. \eqref{ptime}, one can easily
    find the time evolution of the probability
    vector,
 \begin{eqnarray}\label{qubitprob}
   \nonumber p_0&=&\frac{1}{4}\left(1+\e^{-2\left(t_2+t_3\right)}+
   \e^{-2\left(t_1+t_3\right)}+\e^{-2\left(t_1+t_2\right)}\right),
   \\
   \nonumber p_1&=&\frac{1}{4}\left(1+\e^{-2\left(t_2+t_3\right)}-
   \e^{-2\left(t_1+t_3\right)}-\e^{-2\left(t_1+t_2\right)}\right),
   \\
   \nonumber p_2&=&\frac{1}{4}\left(1-\e^{-2\left(t_2+t_3\right)}+
   \e^{-2\left(t_1+t_3\right)}-\e^{-2\left(t_1+t_2\right)}\right),
   \\
   p_3&=&\frac{1}{4}\left(1-\e^{-2\left(t_2+t_3\right)}-
   \e^{-2\left(t_1+t_3\right)}+\e^{-2\left(t_1+t_2\right)}\right),
 \end{eqnarray}
where the total interaction time  reads, $t=\sum_{\mu=1}^3 t_\mu$.

\section{Geometry of the set  $\mathcal{A}_N^Q$ of maps accessible by a semigroup}
\label{sec:three}

The set $\mathcal{A}_N^Q$ of  mixed unitary Weyl quantum channels
acting on $N$-dimensional states and accessible by a
dynamical semigroup is not convex \cite{WECC08}.
In Corollary \ref{logconv} we have shown that
for any dimension $N$  the set ${\cal A}_N^Q\subset {\cal W}_N$
forms a log--convex subset of the probability
simplex  $\Delta_{N^2-1}$ representing  all Weyl mixed unitary channels.
 Another important property of this set is established below.

  \subsection{The set $\mathcal{A}_N^Q$ of accessible maps is star-shaped}

  We show that the log-convex set
$\mathcal{A}_N^Q$ of quantum Weyl channels accessible by a
dynamical semigroup is  star-shaped.

\begin{proposition}
\label{*shap}
The set $\mathcal{A}_N^Q\in\Delta_{N^2-1}$ of mixed unitary
Weyl channels accessible by a semigroup has the
star-shape property with respect to completely depolarizing
channel $\Phi_*$.
\end{proposition}
\begin{proof}
The aim is to show for a quantum channel
$\Phi\in\mathcal{A}_N^Q$ any convex combination with
$\Phi_\ast$ belongs to $\mathcal{A}_N^Q$,
$\Phi'=m\Phi+(1-m)\Phi_\ast\in \mathcal{A}_N^Q$, where $0\leq
m\leq1$. For any $t_\mu\geq0$  associated with
$\Phi\in\mathcal{A}_N^Q$ through Eq. \eqref{lt}, one can
define the non-negative parameter $ t_\mu^\prime=t_\mu-N^{-2}\log{m}$  for which:
\begin{eqnarray*}
\Phi'&=&\e^{ \sum_\mu t^\prime_\mu\mathcal{L}_\mu}=
\e^{ \sum_\mu t_\mu\mathcal{L}_\mu}
\e^{-\frac{\log(m)}{N^2}\sum_\mu\mathcal{L}_\mu}=
\Phi_t \e^{-\frac{\log(m)}{N^2}
\sum(U_\mu\otimes\overline{U}_\mu-\Id_{N^2})}
\\&=&
\Phi_t\e^{\log(m)\Id_{N^2}-\log(m)\Phi_\ast}=
m\Phi_t\e^{-\log(m)\Phi_\ast} = 
m\Phi_t\big((-1+\frac{1}{m})\Phi_\ast+\Id_{N^2}\big)=m\Phi_t+(1-m)\Phi_\ast.
\end{eqnarray*}
Here we use the fact that all subleading eigenvalues of
$\Phi_\ast$ are zero, and that the composition of any unital
channel with the completely depolarizing channel $\Phi_\ast$ is equal to $\Phi_\ast$.
\end{proof}

Note that the boundaries  of the set $\mathcal{A}_N^Q$
consist of channels accessible by a semigroup with
at least one of the non-trivial ($\mu\neq0$)
 parameters $t_\mu$  in Eq. \eqref{lt} vanishing.
This means that all quantum channels of the form
$\exp\left(\sum_{\mu=1}^{N^2-2}\c L_\mu  t_\mu\right)$
belong to  $\partial {\cal A}_N^Q$,
for any choice of $N^2-2$ out of
$N^2-1$ non-trivial Lindblad generators and
those channels using $N^2-1$
Lindblad generators,
 $\Phi_{\vec{p}}=\exp\left(\sum_{\mu=1}^{N^2-1}\c
L_\mu t_\mu\right)$,
lie inside ${\cal A}_N^Q$. To see that,
note that if  $\Phi_{\vec{p}}$ belongs to $\mathcal{A}_N^Q$, it is
connected to $\Id_{N^2}$ through a trajectory indicated by
non-negative interaction time and each point of this
trajectory corresponds to a channel belonging to
$\mathcal{A}_N^Q\in\Delta_{N^2-1}$. This in turn leads us to
the fact that for any $t\geq0$ corresponding to a point on
the trajectory connecting $\Id_{N^2}$ to  $\Phi_{\vec{p}}$,
the dynamical map is positive, also for small time $t$.
Expanding the dynamical map around $t=0$,
positivity condition of the corresponding Choi matrix, $D_{\Phi}
= \Phi^{R}$, implies:
\begin{equation}\label{boundary}
\left(\Id_{N^2}+\epsilon\dfrac{\mathrm{d}\Phi}{\mathrm{d}t}\big\vert_{t=0}\right)^\mathrm{R}\geq0,
\end{equation}
which  holds for any $\epsilon>0$ small enough
 for any $\Phi\in\mathcal{A}_N^Q$. Applying Eq. \eqref{lt}, it
 is easy to see that the above equation is equal to
 $\left(\Id_{N^2}+\epsilon\mathcal{L}\right)^\mathrm{R}=
\left(\Id_{N^2}+\epsilon(\sum_\mu p'_\mu
U_\mu\otimes\overline{U}_\mu-\Id_{N^2})\right)^\mathrm{R}$
which must remain positive
for small $\epsilon$.
 Due to Proposition \ref{cdphi},  its eigenvalues belong to the set
$\{N\left(1-\epsilon(1-p'_0)\right),\epsilon Np'_1,\dots,\epsilon Np'_{N^2-1}\}$ and each of which should be non-negative.
 Noting that $\epsilon$ is small, the first eigenvalue is always
positive.  However,  if $p'_\mu=0$ for any
$\mu\in\{1,2,\dots,N^2-1\}$, the corresponding eigenvalue is
zero which means that  Eq. \eqref{boundary} is saturated,
so that this subspace belongs to the boundary of $\mathcal{A}_N^Q$.
Due to Eq.~\eqref{p'} this fact implies
that  the boundary of $\mathcal{A}_N^Q$ is achieved
if we set to zero the interaction time $t_\mu$  associated with
any selected generator $\mathcal{L}_\mu$.
The above statement is valid for any $N$.
In particular, it implies Proposition (2) shown in \cite{PRZ19} for the case of a single  qubit.

\section{Hyper--decoherence of quantum Weyl channels}
\label{sec:six}

\subsection{Corresponding circulant bistochastic transition matrices}

 As discussed in Sec.~\ref{sec:one}, one can assign a classical stochastic
  transition matrix $T$ to any quantum channel by decohering and reshaping
  the dynamical matrix, $D_{\Phi}=\Phi^R$,
 so that $T$ forms  a minor of order $N$ of the superoperator $\Phi$ of order $N^2$.
  If $\{K_n\}$ represents the set of Kraus operators related to a channel,
 it is straightforward to see the associated stochastic map is gained by
 $T=\sum_nK_n\odot\overline{K}_n$. Thus for the Weyl channels defined by
 Eq.~\eqref{phi} the associated bistochastic transition
 $T$ is given by
 \begin{equation}
 T(\Phi_{\vec{p}})=\sum_{k,l=0}^{N-1}p_{kl}U_{kl}\odot\overline{U}_{kl}=
 \sum_{k,l=0}^{N-1}p_{kl}X^kZ^l\odot X^k\overline{Z^l}=
 \sum_{k,l=0}^{N-1}p_{kl}X^k.
 \end{equation}
 To prove above equation we used the fact that for two arbitrary matrices
$A$ and $C$ and two diagonal matrices $B$ and $D$, the equality
$AB\odot C D=(A\odot C)(B\odot D)$ holds.
Since the operators $Z^l$ do not influence the terms
contributing to the transition matrix $T$,
 it is convenient to introduce an
$N$-point probability vector $\overrightarrow{q}$,
given by the marginal  of the vector $\overrightarrow{p}$ of length $N^2$,
\begin{equation}
\label{q}
 q_k=\sum_l p_{kl}.
\end{equation}
Note that   $\overrightarrow{q}$ can also be considered as a reduction of
$\overrightarrow{p}$,  analogous to partial trace.
Thus  the transition matrix $T_q={\cal D}_h(\Phi_p)$
corresponding the Weyl channel  (\ref{mixeduni})
has the following form
  \begin{equation}
  \label{T}
  T_q(\Phi)=\sum_{k=0}^{N-1}q_kX^k \in {\cal C}_N.
  \end{equation}

The set ${\cal C}_N$ of classical transitions defined by Eq. \eqref{T}
consists of \emph{circulant} bistochastic matrices of order
$N$, i.e., bistochastic matrices whose rows (columns)
are cyclic permutations of its first row (column).
Hence the set ${\cal C}_N$  forms an  $(N-1)$-dimensional
simplex $\Delta_{N-1}=  \{\Id_N, X, X^2, ..., X^{N-1}\}$ embedded
inside the Birkhoff polytope ${\cal B}_N^C$ of dimension $(N-1)^2$
which contains all permutation matrices of order $N$.
Note that the shift operator $X$, also written $X_N$,
is actually the $N$-element cycle permutation,
thus hyper-decocherence of the Weyl channels
gives the set of circulant bistochastic matrices,
${\cal D}_h({\cal W}_N)={\cal C}_N$.

\begin{lemma}\label{com.xk}
Arbitrary powers of cycle permutation do commute,  $[X^k,X^{k'}]=0$.
This implies that all (not necessarily
convex) combinations of $X^k$ are compatible, so the transitions matrices
in Eq.\eqref{T} satisfy relation
$T(\Phi_{\vec{p}})T(\Phi_{\vec{p}'})
=T(\Phi_{\vec{p}'})T(\Phi_{\vec{p}})$.
\end{lemma}

\subsection{Spectral properties of classical and quantum maps }

The Fourier matrix of order $N$, with entries
 $F_{mn}=\exp(-i 2 mn \pi/N)$,
diagonalizes any combination of
powers $X_N^k$ of the cyclic permutation matrix $X_N$.
The common eigenbasis is given by
$\ket{x_j}=\frac{1}{\sqrt{N}} F^\dagger\ket j$. Also eigenvalues and
coefficients of expansion can be transformed to each other applying $F$
and $F^\dagger$. For instance,  for a transition matrix $T$  in Eq. \eqref{T},
its eigenvalues $\xi_i$ and the probabilities $q_k$ satisfy relations,
$\vec{\xi}=F\vec{q}$ and
$\vec{q}=\frac{1}{N}F^\dagger\vec{\xi}$.
To see this let us rewrite $T$ in the basis $\{\ket{x_j}\}$
and use the notation, $\omega_N=\exp(2 \pi i/N)$,
\begin{eqnarray}
\label{fourcirc}
\nonumber\bra{x_m}T\ket{x_n}&=&
\frac{1}{N}\sum_kq_k\sum_x\sum_i\sum_j
\omega_{_N}^{-im}\omega_{_N}^{jn}
\bra i x\oplus k\rangle\bra x j\rangle\\&=&
\frac{1}{N}\sum_kq_k\omega_{_N}^{-mk}\sum_x
\omega_{_N}^{x(n-m)}=\sum_k\omega_{_N}^{-mk}q_k\delta_{mn}.
\end{eqnarray}

As any circulant bistochastic transition matrix $T\in {\cal C}$
contains non-negative real entries and
can be diagonalized by the Fourier matrix $F_N$,
the following spectral properties hold:
\begin{enumerate}[i)]
\item  Any circulant bistochastic matrix $T$ has the leading (Frobenius--Perron) eigenvalue
 $\xi_0=1$.

\item   For the matrix  $T_\ast$, located at the center of the
   Birkhoff polytope ${\cal B}_N$ of bistochastic matrices and corresponding to the flat probability vector,
$q_k=\frac{1}{N}$  for $k=0,\dots,N-1$,
all subleading eigenvalues vanish, $\xi_i=0$ for $i\neq 0$.

\item  For a cyclic permutation matrix  $X_N$ its spectrum belongs to the unit circle
  and forms the $N$--sided regular polygon.

\item  Eigenvalues of an arbitrary $T$ are either real or appear in complex conjugate pairs.

\item  $\xi_{-i\oplus N}=\overline{\xi}_i$,
so for an even $N$ apart from the leading eigenvalue $\xi_0$,
 there exists another real eigenvalue $\xi_{\frac{N}{2}}$.

\item  As the modulus of the entries of $F$ is unity and eigenvalues
 of $T$ are just a convex combination of these elements, so
 the spectrum belongs to the unit disk,
$|\xi_i|\leq1$ for any $i$. Furthermore, as the entries of $F_N$
are given by powers of the $N$-th root of unity, $\omega_N=e^{i 2\pi/N}$,
then all eigenvalues $\xi_i$ belong to the regular $N$--sided polygon
containing  $\xi_0=1$.
\end{enumerate}

The above properties of the circulant bistochastic matrices -- classical maps --
can be extended for certain quantum maps.

\begin{proposition}
\label{prop-channel-hyperdeco}
Consider a Weyl channel $\Phi_{\vec p}$ from the simplex
$\Delta(\Phi_{\mathbb{I}}, \Phi_ X, \dots, \Phi_{X^{N-1}})$.
The set of eigenvalues of the superoperator $\Phi_{\vec p}$ is then equal (up to degeneracy)  to the spectrum of the corresponding circulant bistochastic matrix  $T(\Phi_{\vec p})$.
Moreover, the spectrum of $\Phi_{\vec p}$ is $N$-fold degenerate.
\end{proposition}
\begin{proof}
Let $p_{kl}$ denote the probability vector of size $N^2$  defining  the Weyl
channel $\Phi_{\vec p}$. Due to Proposition~\ref{cdphi}
the eigenvalues of the superoperator $\Phi_{\vec p}$ read
\begin{equation}
    \lambda_{mn} = \sum_{kl} \omega_N^{ml-nk} p_{kl} .
\label{lambdamn}
\end{equation}
Since hyper-decoherence transforms this channel into
the classical map,  $T(\Phi_{\vec p}) = \sum X^k q_k$,
with $q_k = \sum_l p_{kl}$
the spectrum of the circulant bistochastic matrix
$T(\Phi_{\vec p})$ can be expressed
by the Fourier matrix $F$ of size $N$,
\[
\xi_n = \sum_k F_{nk} q_k,
\]

\noindent
Thus Proposition~\ref{prop-channel-hyperdeco} is equivalent to showing that $\forall m$, $\lambda_{mn} = \xi_n$.
We start with  expression (\ref{lambdamn})
 for eigenvalues $\lambda_{mn}$
 and notice that since  $\Phi_{\vec p}$ belongs to the
 the simplex
$\Delta(\Phi_{\mathbb{I}}, \Phi_ X, \dots, \Phi_{X^{N-1}})$
and $p_{kl} \neq 0$ only for $l = 0$, then
\begin{align*}
    \lambda_{mn} &= \sum_{kl} \omega_N^{ml-nk} p_{kl} = \sum_k \omega_N^{-nk} p_{k0} \\
    &= \sum_k \omega_N^{-nk} q_k = \sum_k F_{nk} q_k = \xi_n.
\end{align*}
In the second row we used that $q_k = \sum_l p_{kl} = p_{k0}$ which is due to the fact that
$\Phi_{\vec p}$ belongs to the simplex $\Delta$.
This shows that both set of eigenvalues are equal.
\end{proof}
\begin{proposition}
  Consider an arbitrary Weyl  channel $\Phi_{\vec p}$.
  There exists a change of  basis  such that  $\Phi_{\vec p}$ is
  block diagonal, each block of order $N$ forms a circulant
  matrix.  Thus the spectrum of the superoperator $\Phi_{\vec p}$
  is contained in the regular $N$-polygon, with its rightmost vertex being $1$.
  In addition, one of these blocks forms matrix $T(\Phi_{\vec p})$,
  so the spectrum of  $\Phi_{\vec p}$ contains the spectrum of the
  corresponding transition matrix  $T$.
\end{proposition}
\begin{proof}
  Consider the basis given by $\{\ket{j\oplus k,j}\}$
  for  $k$, $j\in\{0,\dots, N-1\}$,
  corresponding to a permutation of the original basis.
  The next step is to compute an entry of $\Phi_{\vec p}$ in
  the new basis:
  \begin{align*}
    \bra{j\oplus k, j}
    \Phi_{\vec p}
    \ket{j'\oplus k', j'}
    &=
    \bra{j\oplus k, j}
    \bigl(
      \sum_{mn} p_{mn} U_{mn}\otimes \overline{U}_{mn}
    \bigr)
    \ket{j'\oplus k', j'}
    \\
    &=
    \sum_{mn}
    p_{mn}
    \bra{j\oplus k}
    U_{mn}
    \ket{j'\oplus k'}
    \bra{j}
    \overline{U}_{mn}
    \ket{j'}\\
    &=
      \sum_{mn} p_{mn}
    \bra{j\oplus k}
    X^{m}Z^{n}
    \ket{j'\oplus k'}
    \bra{j}
      \overline{X^{m}}
      \overline{Z^{n}}
    \ket{ j'}
    \\
    &=
    \sum_n p_{-j'\oplus j, n} \omega_N^{nk}\delta_{k,k'}.
  \end{align*}

  Thus the following decomposition holds
  \begin{equation}
    \Phi_{\vec p} =
    \sum_{k=0}^{N-1}
    \sum_{j,j'=0}^{N-1}
    \Bigl(
    \sum_l p_{-j'\oplus j, l} \omega_N^{kl}
      \Bigr)
      \ket{j\oplus k, j}
      \bra{j'\oplus k, j'}.
    \label{eq-block-diagonal-quantum-channel}
  \end{equation}
  Fixing $k$, the remaining sum on $j$ and $j'$ corresponds
  to a block $\Phi_{\vec p}^{(k)}$ of size $N$, given by
   \begin{equation}
    \Phi_{\vec p}^{(k)} =
   \sum_{k'}^{N-1}
   \Bigl(
     \sum_{l=0}^{N-1}
     p_{k' l}
     \omega_N^{kl}\Bigr)
     X^{k'}
     =
     \sum_{k'}
     X^{k'}
     \Bigl(
       \sum_{ml} (Z^{k})_{ml} p_{k' l}
       \Bigr) =
     \sum_{k'}
     X^{k'}
     q_{k'}^{Z^{k}},
  \end{equation}
  where we introduced a shorthand notation,
  $q^{Z^{k}}_m  \equiv
\sum_{lk'} (Z^{k})_{k'l} p_{ml}$.
Therefore, we can write
\[
  \Phi_{\vec p}
  = \bigoplus_{k}
  \Phi_{\vec p}^{(k)}.
\]
\par
Given that each block $\Phi_{\vec p}^{(k)}$ is circulant and can be
diagonalized by the $N$ dimensional Fourier matrix,
 the spectrum of  $\Phi_{\vec p}$ is contained in the regular
$N$-polygon.
Finally, observing that $q^{Z^{0}}_m = q_m$,
we realize that the first block of the superoperator is
equal to its classical analogue,
$\Phi_{\vec p}^{(0)} = T(\Phi_{\vec p})$.
\end{proof}

\subsection{Kolmogorov generators and accessible bistochastic matrices}

In analogy with the quantum case, it is interesting to find
 which transition matrices out of those described in \eqref{T} are accessible
by a classical dynamical semigroup.
From a mathematical perspective such stochastic matrices are also
called \emph{embeddable} and their properties were analyzed in \cite{Ru62,Da10}.
In our physics--oriented approach it will be convenient to make use of the notion of
\emph{strictly incoherent operations} (SIO), defined as these
stochastic maps, which are not able
 to generate and to use quantum coherence \cite{WY16,YMGGV16}.
Therefore an arbitrary quantum operation is  strictly incoherent
if and only if  its Kraus operators are of  the form,
$K_n=\sum_i c_{ni}\outprod{P_n(i)}i$
where $P_n(i)$ is a permutation matrix \cite{CG16} and $c_{ni}$
are  arbitrary complex coefficients  satisfying the  normalization condition.

In the analyzed case of mixed Weyl channels
the corresponding Kraus operators read,
$K_{kl}=\sqrt{p_{kl}}\sum_m
\omega_{_N}^{ml}\outprod{m\oplus k}m$,
which proves  that these channels are strictly incoherent.
This property is  equivalent to  the following decoupling property,
\begin{equation}
\Phi_{\vec{p}} \left(\c D\left(\rho\right)\right)=\c D\left(\Phi_{\vec{p}}\left(\rho\right)\right).
\end{equation}
It states that the decoherence map $\c D$, acting in the space of density matrices --
see Sec.~\ref{sec:one} --
 is not coupled with the dynamics generated by a Weyl channel
$\Phi_{\vec{p}}$.
Thus coherence  of input states of  Weyl channels  plays
no role in the evolution of probability vector of their diagonal entries.
Applying  a mixed unitary Weyl channel \eqref{phi} on
 an input state $\rho$ we see
 that the evolution of its diagonal entries
 is governed by the classical transition matrix  $T$ given in  \eqref{T}
 and related
 to the Weyl channel through Eq. \eqref{supermap}.
If $d_Y$ denotes the vector of diagonal
elements of a matrix $Y$, written $d_Y=\diag(Y)$,
then the following equality holds,
\begin{equation}
Td_\rho=d_{\Phi(\rho)}.
\end{equation}
This relation, not working for an arbitrary quantum channel,
 shows that in the case of a Weyl channel, the
 corresponding transition matrix  $T(\Phi)$
 is responsible for the evolution of the diagonal elements of $\rho$.
This fact can be restated in the following way.

\begin{proposition}
\label{pauli master prop.}
The evolution of diagonal elements of input states for the channels
defined by Eq. \eqref{phi} that also are accessible by a dynamical
semigroup, i.e. those Weyl channels describing a Markovian evolution, is
governed by Markovian Pauli master equation:
\begin{equation}\label{pauli master}
\frac{\rm{d}}{\rm{dt}}d_\rho=\c Kd_\rho,
\end{equation}
where $\c K$ is the Kolmogorov operator.
\end{proposition}
 So $d_{\rho(t)}$ in the above equation is determined as
 $d_\rho(t)=d_{\rho(t)}=\e^{t\c K}d_\rho(0)=\mathbb{T}_td_\rho(0)$.
 For future references let us explicitly mention necessary properties
 of a Kolmogorov operator  represented by a real matrix of order $N$ -- see \cite{CMMV13},

 a)$\forall i\neq j \ \ \ \c K_{ij}\geq0$,

 b)$\forall j \ \ \ \ \ \  \sum_i \c K_{ij}=0$, or equivalently:

 b')$\forall j \ \ \ \ \ \ \ \c K_{jj}=-\sum_{i\neq j}\c K_{ij}$.

\medspace

As discussed in Section \ref{sec:one},
the process of decoherence  $\c D$ brings a quantum state
into a classical probability vector,
while the hyper-decoherence  $\c D_{h}$
sends a quantum channel $\Phi$ into a classical
 transition matrix $T$.
 In a similar manner a Lindblad operator $\c L$ governing the quantum dynamics
 is transformed by process of hyper-decoherence into a Kolmogorov operator $\c K$
 generating a classical semigroup.
 Consider first a Lindblad generator determining a quantum semigroup
 \eqref{p'} related to accessible Weyl channels.
 \begin{lemma}\label{K to L}
Decohering and then reshaping the diagonal elements of  $\c L^R$,
such that
$\c L$ is defined in Eq. \eqref{p'},  gives a proper Kolmogorov
operator corresponding to the Lindblad generator.
 \end{lemma}
 \begin{proof}
  It is straightforward to see the aforementioned method results in the
  following operator,
 $\c K=\sum_{k,l}p'_{k,l}U_{kl}\odot\overline{U}_{kl}-\Id_N$.
 Thus based on the discussions right after Eq.\eqref{T}, we have:
 \begin{equation}\label{kolm}
  \c K=\sum_{k=0}^{N-1}q'_kX^k-\Id_N=\sum_kq'_k\c K_k
 \end{equation}
 where $q'_k=\sum_lp'_{kl}$ and $\c K_k=X^k-\Id_N$ is the
 Kolmogorov operator associated to $\c L_{kl}$ and is independent of
 $l$. Note that $\c K_0=0$ is the trivial case. It is easy to see $\c K$ and
 $\c K_k$  satisfy properties (a), (b) and (c) of a
 valid Kolmogorov operator which completes the proof.
 \end{proof}

In general, the process of hyper-decoherence acting in the space of Lindblad operators
produces a valid Kolmogorov generator  \cite{Ko72,CMMV13},
with entries ${\c K}_{ij} = {\rm Tr}[|i\rangle \langle i|\; {\c L}(|j\rangle \langle j|)] $.
Thus any Lindblad operator produces,  by the $N$--dimensional  projection operator $\Pi_N$
used in  Eq.~\eqref{supermap},
a corresponding classical  Kolmogorov generator,
 \begin{equation}
 {\c K} \left({\c L}\right)  = \c D_{h}(\c L) =  \Pi_N {\c L} \Pi_N , \  \ {\rm and \ \ }
  {\c K}_{ij} ={\c L}_{\substack{ii\\ jj}} .
 \label{Lind-Kol}
 \end{equation}
However, the analyzed case of Weyl channels and the corresponding Lindblad operators
is rather special, as all the Lindblad  generators
are commutative - see Eq. (\ref{lt}).
This property is also inherited in the classical setup,
as due to Lemma \ref{com.xk}
 all the generators  $\c K_k$ are commutative,
 so are all $\c K$ and $\c K'$ in the form of  Eq.
 \eqref{kolm}. The commutativity of Kolmogorov operators $\c K_k$
 imposes commutativity on their respective dynamical semigroups,
 $\mathbb{T}_{t_k}\mathbb{T}_{t_{k'}}=
 \mathbb{T}_{t_{k'}}\mathbb{T}_{t_k}$ where
 $\mathbb{T}_{t_k}=\e^{t_k\c K_k}$.

Moreover,  classical transition matrices  gained by any concatenation of classical
 dynamical semigroups, $\e^{t_k\c K_k}$,
 belong to the set ${\cal C}_N$ of circulant bistochastic
 matrices  (\ref{T}), i.e.
 \begin{equation}\label{classical dynamical semigroups}
\mathbb{T}_t=
\mathbb{T}_{t_0}\mathbb{T}_{t_1}\dots\mathbb{T}_{t_{N-1}}=
\e^{t_0\c K_0}\e^{t_1\c K_1}\dots\e^{t_{N-1}
\c K_{N-1}}=\e^{\sum_{k=0}^Nt_k\c K_k}=\e^{t\sum q'_k
\c K_k}=\e^{t\c K}=T_q,
 \end{equation}
 where the total interaction time $t=\sum t_k$ and $q'_k={t_k}/{t}$.
 Observe that  the order of
 classical dynamical semigroups is not important because of their
 commutativity. To see that, one can apply the Fourier matrix $F$ to diagonalize the
 Kolmogorov operator and to find eigenvalues, so $\mathbb{T}_t$ has
 the following form:
  \begin{equation}\label{spect. of dyn. semi. gr.}
  \mathbb{T}_t=\sum_j\e^{t(\xi'_j-1)}\project{x_j},
  \end{equation}
 where $\xi'_j$ are the eigenvalues of $\sum q'_kX^k$
 obtained by applying $F$ on the probabilities $q'_k$. Therefore,
 the above classical dynamical semigroup
  can be rewritten in the form \eqref{T} in
 which probabilities $q_k$ are given by,
\begin{equation}\label{qk1}
q_k=\frac{1}{N}\sum_jF^\dagger_{kj}\e^{t(\xi'_j-1)}=
\frac{1}{N}\sum_{j}\omega_{_N}^{kj}\exp{\left(\sum_{m}
(\omega_{_N}^{-jm}-1)t_{m}\right)}.
\end{equation}

 In Theorem \ref{IM A_N^Q}  we will prove  that this equality provides
 a valid probability vector with non-negative entries.
Observe that according to Eq. \eqref{spect. of dyn. semi. gr.},
 for any real $\xi'_j$, the corresponding eigenvalue of
 $\mathbb{T}_t$ is positive.
 For the complex conjugate pair $\xi'_j$ and
 $\xi'_{-j\oplus N}$, the corresponding eigenvalues of
 $\mathbb{T}_t$ are complex conjugates as well.
Moreover, at $t_+=\frac{2n\pi}{|\rm{Im}(\xi'_j)|}$ and
$t_-=\frac{(2n-1)\pi}{|\rm{Im}(\xi'_j)|}$ for any $n\in\mathbb{N}$,
$\mathbb{T}_t$ will possess positive degenerate and negative
degenerate eigenvalues respectively. Finally, as
 $|\xi'_j|\leq 1$, the moduli of the eigenvalues of
 $\mathbb{T}_t$ are either unity or
 decay exponentially with respect to $t$.

 We denote by $\c A_N^C\subset {\cal C}_N$ the set of circulant bistochastic transitions
 satisfying Eq. \eqref{classical dynamical semigroups} and forming the set of
 transition matrices accessible by a classical dynamical semigroups with
 the Kolmogorov operator in the form of Eq. \eqref{kolm}. The
 aforementioned equation also implies the following result,
 analogous to the quantum counterpart
formulated in Corollary \ref{logconv},

 \begin{corollary}
 The set $\c A_N^C$ of classical transition matrices
 acting on $N$-point probability vectors
 and  accessible by a classical dynamical
 semigroup generated by $\c K$ in Eq. \eqref{kolm} is
 $\rm{log\ convex}$.
 \end{corollary}

\begin{theorem}
Assume now a circulant bistochastic transition in the form of
\begin{equation}
T=\sum_n\xi_n\project{x_n}.
\end{equation}

Again, we choose logarithm such that
$\overline{\log\xi_n}=\log\overline{\xi_n}=\log\xi_{-n\oplus N}$.
The necessary and sufficient
condition for $T$ to be accessible by a classical dynamical
semigroup  \eqref{classical dynamical semigroups},  $T\in\c A_N^C$,
is that  the sum on the right hand side of Eq. (\ref{NS for T}) is real and positive,
\begin{equation}
\label{NS for T}
t_m=\frac{1}{N}\sum_n(QF^\dagger Q)_{mn}
\log\xi_n\geq0,\
\end{equation}
for $ m=1,\dots, N-1$.
Here $F$ denotes  the Fourier matrix, $Q=\Id_N-\project{0}$   is the
projector operator on $(N-1)$-dimensional subspace and
$\project{0}$ projects $\vec{\xi}$ to $\xi_0=1$.
Note that there is no condition on $t_0$ as it is the coefficient of
 $\c K_0=0$.
\end{theorem}

\begin{proof}
Let us first show that if $T$ is an accessible classical map, then
Eq.~\eqref{NS for T} holds. So we  start by the dynamical
semigroup in Eq. \eqref{spect. of dyn. semi. gr.}
that implies the following relation for the
eigenvalues of an accessible map:
\begin{equation}\label{e1}
\log\xi_n=t(\xi'_n-1)=t(\sum_{m=0}^{N-1}
\omega_{_N}^{-mn}q'_m-1)=
\sum_{m=0}^{N-1}(\omega_{_N}^{-mn}-1)t_m=
\sum_{m=0}^{N-1}F'_{mn}t_m,
\end{equation}
where $F'=F-V$ with $V_{mn}=1$. As the entries of the first row and
column of  $F'$ are all zero, $F'$ is not invertible and $F'=QF'Q$. It is
again possible to define the inverse on the $(N-1)$-dimensional
subspace $Q$.
In the following we show this inverse exists and equals to
$QF^\dagger Q$, up to the constant $N$:
\begin{eqnarray}
\nonumber(QF^\dagger Q)(QF'Q)&=&
\sum_{m,n=1}^{N-1}\sum_{i,j=0}^{N-1}
\sum_{k,l=0}^{N-1}\omega_{_N}^{ij}(\omega_{_N}^{-kl}-1)
\project mi\rangle\langle j\project n k\rangle\bra l\\ \nonumber
&=&\sum_{m,n=1}^{N-1}\sum_{l=0}^{N-1}
\omega_{_N}^{mn}(\omega_{_N}^{-nl}-1)\outprod ml\\
&=&N\sum_{m=1}^{N-1}\sum_{l=0}^{N-1}\delta_{ml}\outprod ml
=N\sum_{m=1}^{N-1}\outprod mm= N(\Id_N-\project{0}).
\end{eqnarray}
So by acting $(QF^\dagger Q)$ on the both sides of Eq. \eqref{e1} the
condition in Eq. \eqref{NS for T} is achieved for interaction time $t_m$
which should be non-negative. Moreover, note that Eq.~\eqref{e1}
implies $\overline{\log\xi_n}=\log\overline{\xi_n}$.
On the other hand, if the inequality \eqref{NS for T} holds between
eigenvalues of a classical map, as the interaction time
is explicitly presented, it is clear that the map belongs to $\c A_N^C$.
\end{proof}

As in the quantum version,  condition~\eqref{NS for T} for $T$ implies
existence of a proper logarithm of $T$ satisfying
 properties of a Kolmogorov operator mentioned
right after Prop. \ref{pauli master prop.}. In analogy to the quantum case
one can simply write $\log T$ as follows:
\begin{equation}
\log T= \sum_{n=0}^{N-1}\log\xi_n\project{x_n}=\frac{1}{N}
\sum_{m=0}^{N-1}\sum_{n=0}^{N-1}\sum_{j=0}^{N-1}
\omega_{_N}^{mn}\log\xi_n\outprod{m\oplus j} j.
\end{equation}
For this generator   $\forall j: \
\sum_i(\log T)_{ij}=0$ is straightforwardly verified.
In addition,   $\forall i\neq j:\ (\log T)_{ij}\geq0 $   means
\begin{eqnarray}
 \forall m\neq0:\quad\quad\quad \frac{1}{N}
 \sum_{n=0}^{N-1}\omega_{_N}^{mn}\log\xi_n\geq0.
\end{eqnarray}

Interestingly, the interaction times  $t_m$ entering Eq. \eqref{NS for T}
can be obtained by
summation over the second index of $t_{mn}$ introduced in Eq.~\eqref{time}
when it is rewritten by two indices $m,n$ rather than
$\mu$. This observation suggests a closer relation between
quantum accessible channels and classical accessible transition matrices.
 The following theorem
shows that $\c A_N^C$ is the image of $\c A_N^Q$ in classical space,
 i.e. $\c D_{h}(\c A_N^Q)=\c A_N^C$,
 where hyper-decoherence $\c D_{h}$ is defined by the projection \eqref{supermap}.

\begin{theorem}\label{IM A_N^Q}
Let $\Phi_{\vec{p}}$ be a Weyl channel accessible by a Lindblad generator in
Eq. \eqref{p'}, i.e. $\Phi_{\vec{p}} =\e^{t\c L}\in\c A_N^Q$.
Then the corresponding
circulant bistochastic matrix $T(\Phi)$ belongs to the set of
transition matrices  accessible by a classical semigroup whose generator $\c K$ is
defined by  Eq. \eqref{kolm} based on $\c L$, i.e.
$T=\e^{t \c K}\in\c A_N^C$ and $\c K$ is related to $\c L$
 through Lemma \ref{K to L}. Hence for a generator $\c L$ related to a Weyl channel the
 hyper-decoherence commutes with the time evolution,
 \begin{equation}
 \c D_{h}\left(\e^{t\c L}\right)=\e^{t\c D_{h}\left(\c L\right)}
  =e^{t {\c K}}.
 \end{equation}
\end{theorem}
\begin{proof}
We start the proof by finding classical transition of $\Phi=\e^{t\c L}$.
Assuming $t_{mn}=tp'_{mn}$, and rewriting Eq.~\eqref{ptime}
using two indices $k,l$ one has:
\begin{equation}
\Phi_{\vec{p}} =\e^{t\c L}=\sum_{k,l}p_{k,l}U_{kl}\otimes\overline{U}_{kl},
\end{equation}
where
\begin{equation}
p_{kl}=\frac{1}{N^2}\sum_{ij}\omega_{_N}^{kj-li}
\exp{\left(\sum_{mn}(\omega_{_N}^{in-jm}-1)t_{mn}\right)}.
\end{equation}
Hence $T$ can be gained by Eq. \eqref{T}, in which the probabilities $q_k$ read,
\begin{eqnarray}\label{qk2}
\nonumber q_k&=&\sum_l p_{kl}=\frac{1}{N^2}
\sum_{ij}\omega_{_N}^{kj}\sum_{l}(\omega_{_N}^{-li})
\exp{\left(\sum_{mn}(\omega_{_N}^{in-jm}-1)t_{mn}\right)}\\
&=&\frac{1}{N}\sum_{j}\omega_{_N}^{kj}
\exp{\left(\sum_{m}(\omega_{_N}^{-jm}-1)t_{m}\right)},
\end{eqnarray}
where $t_m=\sum_nt_{mn}=t\sum_np'_{mn}=tq'_m$,
while  $q'_m$ is defined by Eq. \eqref{kolm}.
Equalities \eqref{qk1} and  \eqref{qk2} complete the proof.
Furthermore,  Eq. \eqref{qk1} explicitly  represents a  valid probability vector $q$
with non-negative components, as it was claimed before.
\end{proof}

Thus the classical image of a quantum accessible Weyl channel
 is a circulant bistochastic matrix accessible by a classical semigroup.
 An immediate consequence of this fact and Prop. \ref{*shap}  is the following Corollary.

\begin{corollary}
 The set $\c A_N^C$ of circulant bistochastic transition matrices,
   accessible by a semigroup  generated by Kolmogorov operators \eqref{kolm}, is
star-shaped with respect to the flat matrix, $T_\ast=\frac{1}{N}\sum_kX^k$,
in analogy to its quantum counterpart $\c A_N^Q$.
\end{corollary}
Already we have seen the set of classical accessible transitions are
\emph{log-convex} and \emph{star-shaped}. The following Proposition shows how
one can get the boundaries of the set $\c A_N^C$.

\begin{proposition}
\label{prop14}
The boundaries of the set, $\partial\c A_N^C$, are accessible when at
least one of the probabilities $q'_k$ related to non-trivial Kolmogorov
operators ($k\neq0$) \eqref{kolm} is zero. In other words, if one takes at most $N-2$
out of $N-1$ non-trivial Kolmogorov operators, the resultant transitions,
$\e^{t\c K}$,  form the boundaries $\partial\c A_N^C$.
\end{proposition}
\begin{proof}
To prove the claim, note that if $T$ belongs to $\c A_N^C$ it is
 connected to $\Id_{_N}$ through a trajectory that at each
 $t\geq 0$ belongs to the set, so does for small time $t=\epsilon$.
 Expanding $\e^{t\c K}$ around $t=0$, one has:
 \begin{eqnarray}
 \nonumber \e^{\epsilon\c K}&=&\left(\Id_{_N}+\epsilon
 \dfrac{\mathrm{d}T}{\mathrm{d}t}\big\vert_{t=0}\right)=
 \left(\Id_{_N}+\epsilon\c K\right)=
 \left(1-\epsilon(1-q'_0)\right)
 \Id_{_N}+\epsilon\sum_{k=1}^{N-1}q'_kX^k,
 \end{eqnarray}
that should describe a proper bistochastic transition for small
$\epsilon>0$. This means all of its entries should be non-negative and
sum of them over each row and column equals to unity. The latter is
satisfied automatically. The first condition means
$1-\epsilon(1-q'_0)\geq0$ and $\epsilon q'_k\geq0$ for
$k\in\{1,\dots,N-1\}$. Regarding that $\epsilon$ is small
$1-\epsilon(1-q'_0)$ is always positive, however, if $q'_k$ is taken
to be zero for $k\in\{1,\dots,N-1\}$, the second inequality is satisfied
with equality. So in that case we get the boundaries of the set.
\end{proof}

\subsection{Spectra of accessible quantum Weyl channels and circular
bistochastic matrices }
\label{AppSpiral}
In this subsection we analyze the boundary of certain cross-sections
of the set ${\cal A}_N^Q$ of accessible Weyl channels
and characterize the support of their spectra in the complex plane.
As Proposition~\ref{prop-channel-hyperdeco} relates the eigenvalues of a superoperator $\Phi_p$,
corresponding to a Weyl channel, with eigenvalues of its classical action,
described by a circulant bistochastic matrix, our results concern also
support of spectra of transistion matrices accessible by a classical semigroup.
Circulant bistochastic matrices of order N
are spanned by the cycle permutation matrix X and its powers,
$\Delta(X,X^2,..., X^N)$, where $X^N = \mathbb{I}_N$.

We consider the spectrum of accessible Weyl channels
$\c A_N^Q$, \emph{i.e.}, for which
there exists a trajectory, $z(t)$, such that $z(t=0)=\Phi_\Id$ 
 and belongs to the set of Weyl channels at all times $t$.
This is equivalent to the statement that the trajectory in the space  of transition
matrices corresponding   to $z(t)$ by hyper-decoherence
belongs to the simplex $\Delta$ for all times $t$.
In order to obtain an analytical description, it suffices
to study the behavior of the curve at times $t$ close to $0$.
Using this observation, we get the boundary of spectrum $\c A_N^Q$ by a parametric
expression,
\begin{equation}
e^{\pm i t - t \tan\left(\pi/N\right)}
\ \text{ for }
t \in \left[0,\pi \right].
\label{spira}
\end{equation}
\begin{figure}[ht]
	\centering
	\includegraphics[width=0.57\linewidth]{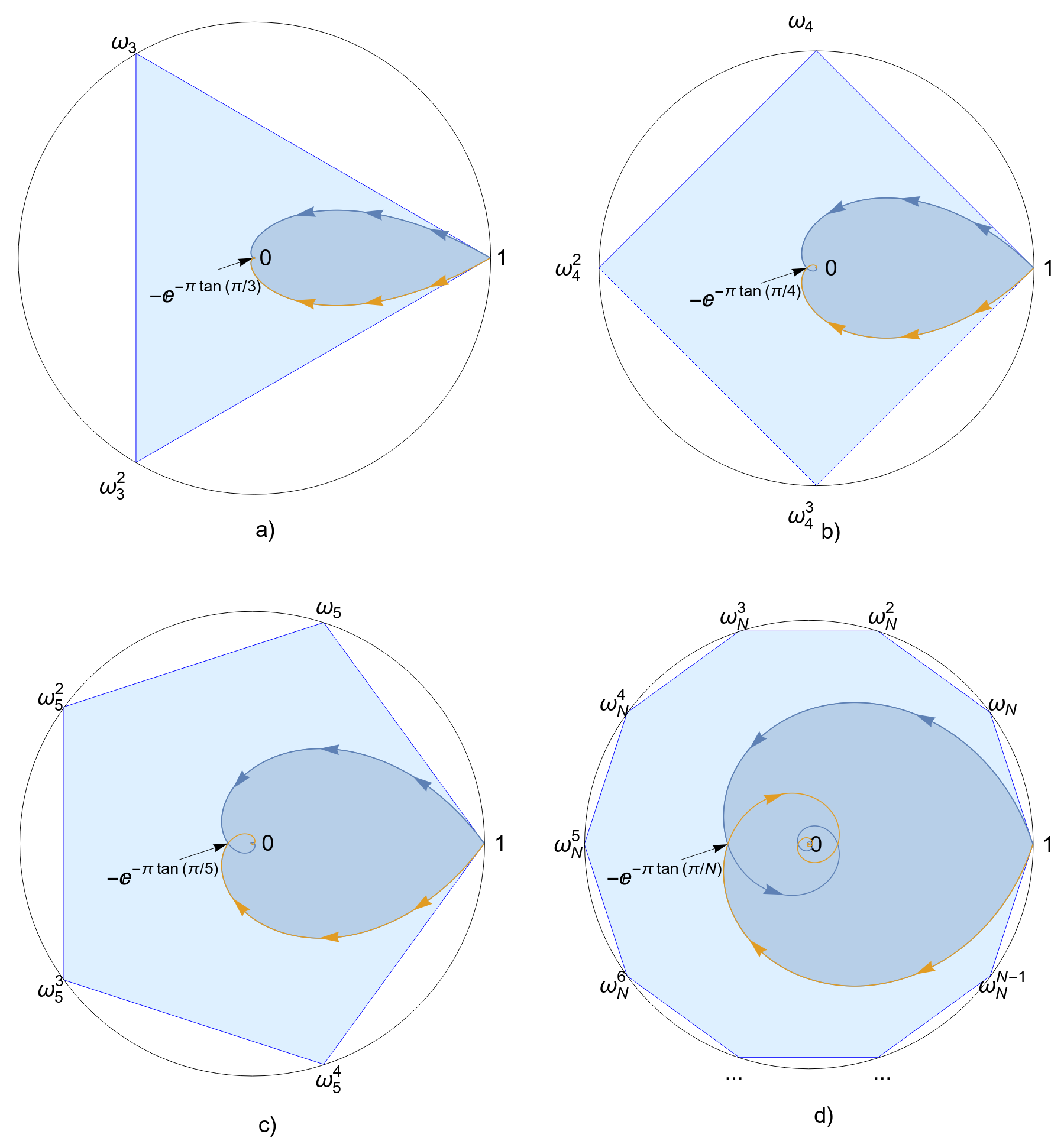}
	\caption{
	  Shaded polygons, spanned by roots of unity, represent
		  support of the spectra of  superoperators $\Phi_p$ representing all Weyl channels
		     acting in dimension $N=3$, $4$, $5$, $10$ for panels a)-d), respectively.
     Support of the spectra of accessible Weyl channels from $\mathcal{A}_N^Q$
      is denoted by a dark set bounded by two logarithmic	spirals
		(\ref{spira}),  which cross at $x_{min}= - e^{-\pi \tan 	\left(\pi/N\right)}$,
		so that the origin, $z=0$, belongs to its interior.
		The same shapes correspond to spectra of all circulant bistochastic
		matrices and matrices accessible by a classical semigroup from ${\cal A}_3^C$.
	}
	\label{fig:accesible}
\end{figure}

Note that in the region described above, equal to the support of spectra
of accessible classical transition matrices from $\mathcal{A}_N^C$, 
there is a line segment of negative values. 
More precisely, the intersection of 
of the spectra of accessible classical and quantum maps
with the real line is
$
\mathcal{A}_N \cap \mathbb{R} = [-e^{-\pi \tan\left(\pi/N\right)}, 1].
$
Thus, the set $\mathcal{A}_N^C$ of bistochastic matrices accessible by a classical semigroup
contains also matrices with degenerated eigenvalues which are real and negative.
The minimum real eigenvalue for a matrix $T$ belonging to $\mathcal{A}^{C}_N$,
or for a superoperator $\Phi$ from $\mathcal{A}^{Q}_N$,
reads
\begin{equation}\label{key}
x_{\rm min}(N)= -e^{-\pi  \tan \left(\pi / N\right)} =
-1+\frac{\pi ^2}{N}+\mathcal{O}\left(\frac{1}{N^2}\right).
\end{equation}
For $N=3$ this number is close to zero, $x_{\rm min}(3)=
-e^{-\sqrt{3}\pi} \simeq -0.004333$,
but for large dimension one has, $x_{\rm min}(N) \to -1$,
and the support of the spectrum of accessible maps covers completely the unit disk.
In order to calculate the area of this set, we consider real and imaginary part
of the parametric expression, $x(t)$ and $y(t)$, respectively
\begin{equation}
\begin{split}
x(t) &=
\cos (t) e^{-t  \tan\left(\pi/N\right)},
\\
y(t) &=
\sin (t) e^{-t  \tan\left(\pi/N\right)}.
\end{split}
\end{equation}
Next, we calculate the area enclosed by this curve using standard analytical methods
\begin{equation}
  V_N=
2 \int_{0}^{\pi} x(t) y'(t) \dd t
=
\frac{1}{2} \left(1-e^{-2 \pi  \tan \left(\pi/N\right)}\right) \cot
\left(\frac{\pi }{N}\right)
=
\pi -\frac{\pi ^3}{N}+\mathcal{O}\left(\frac{1}{N^2}\right) ,
\end{equation}
which for $N=3$ gives
 $V_3=\frac{1-e^{-2 \sqrt{3} \pi }}{2 \sqrt{3}} \simeq 0.28867$.
Observe that for any $z\in \mathbb{C}\backslash\mathbb{R}$,
with $|z| \leq 1$, the curve $z^t$ tends to $0$ as $t \to
\infty$. The trajectory is a logarithmic spiral (\emph{Spira mirabilis}),
 and therefore, it crosses the negative real axis infinitely many times.

It should be mentioned here that in the mathematical
literature, stochastic matrices accessible by a classical semigroup are called
\emph{embeddable} and an expression equivalent to Eq.~\eqref{spira} appeared in the
literature \cite{Ru62,Da10}.
The characteristic heart-like set plotted in Fig.~\ref{fig:accesible}a
was analyzed recently \cite{LKM17} in the context of
Markovian evolution of quantum coherence for $N=3$ systems.
Note that the logaritmic spiral ~\eqref{spira} plotted in the case $N=3$
in Fig. \ref{fig:accesible}a, is relevant for our study in several ways:
 it characterizes
a) the boundary of the set of quantum accessible Weyl channels $\mathcal{A}^{Q}_3$
at the triangular face $\Delta(\Phi_{\mathbb{I}}, \Phi_ X, \Phi_{X^{2}})$ -- see
Fig. \ref{fig:4panel}a; \
b)  the boundary of  the set $\mathcal{A}^{C}_3$ of classical accessible matrices
in the triangle  $\Delta(X,X^2, \mathbb{I}_3)$ of circular bistochastic matrices -- see Fig. \ref{DAfig1}c;
 the boundary of the support of the spectrum of members of
  c)  quantum accessible maps from $\mathcal{A}^{Q}_3$;\
  d) classical accessible transition matrices from $\mathcal{A}^{C}_3$ -- see Fig. \ref{DAfig3}c.
 The two latter statements concerning the boundary of the spectra
   in the complex plane are also valid for higher dimensions $N$
   -- see Fig.   \ref{fig:accesible}b-d.

\section{Lindblad dynamics for qutrits}\label{sec:four}

  \subsection{Explicit criterion for accessibility by a Lindblad semigroup}

Let us now apply previous results for the special case $N=3$.
Although the explicit form of Weyl unitary matrices is not relevant here,
 we wish to explain the notation for future reference.
For the sake of brevity, we are going to use one
index $\mu\in\{0,\dots,8\}$ instead of two $k,l\in\{0,1,2\}$, related by
$\mu=3k+l$.
This means,
up to a phase, $U_1$ is equal to  $U_2^\dagger$,   $U_3$ is
$U_6^\dagger$, $U_4$ equals to $U_8^\dagger$, and  $U_5$ is
$U_7^\dagger$ which in turn implies  following relations on
the eigenvalues:  $\lambda_2=\overline\lambda_1$,
$\lambda_3=\overline\lambda_6$, $\lambda_4=\overline\lambda_8$, and
$\lambda_5=\overline\lambda_7$.

In the case $N=3$ the Hermitian matrix $H$ of size $N^2=9$,
 defined in Eq.~\eqref{M}
can be represented in the element-wise logarithm notation
useful for complex Hadamard matrices \cite{TZ06},
\begin{eqnarray}\label{Mqutrit}
\log H&=&
\frac{2 \pi  i}{3}
\left(
\begin{array}{ccccccccc}
 \bullet & \bullet & \bullet & \bullet & \bullet & \bullet & \bullet & \bullet & \bullet \\
 \bullet & \bullet & \bullet & 2& 2 & 2 & 1 & 1 & 1 \\
 \bullet & \bullet & \bullet & 1 & 1 & 1 & 2 & 2 & 2 \\
 \bullet & 2 & 2 & \bullet & 1 & 2 & \bullet & 1& 2 \\
 \bullet & 1 & 2 & 2 & \bullet & 1 & 1 & 2 & \bullet \\
 \bullet & 1 & 2 & 1 & 2 & \bullet & 2 & \bullet & 1 \\
 \bullet & 2 &  1 & \bullet & 2 & 1 & \bullet & 2& 1 \\
 \bullet & 2 &  1 & 2 & 1 & \bullet & 1 & \bullet & 2 \\
 \bullet & 2 &  1 & 1 & \bullet & 2 &
 2 & 1 & \bullet \\
\end{array}
\right),
\end{eqnarray}
where
 $\bullet$ represents zero,
 so that $H_{11}=\exp(0)=1$, while $H_{24}=\exp(i 4\pi/3)=\omega_3^2$.
Applying $H$ on the probability vector $p_\mu$ defining the
mixed unitary channel (\ref{phi})
yields eigenvalues $\lambda_{\mu}$ of the superoperator $\Phi_{\vec p}$.
To describe a legitimate quantum channel
accessible by a semigroup, $\Phi_{\vec p} \in \mathcal{A}_3^Q$,
the interaction times need to be non-negative, $t_{\mu}\ge 0$,
which  due to Eq.~\eqref{time} yields
constraints for the complex eigenvalues $\lambda_{\mu}$.

\begin{eqnarray}
\nonumber\log\Bigl(\frac{|\lambda_1|^2}{|\lambda_3||\lambda_4||\lambda_5|}\Bigr)\pm\sqrt{3}\left(\theta_3+\theta_4+\theta_5
+M_3+M_4+M_5\right)&\geq&0,\\ \nonumber
 \log\Bigl(\frac{|\lambda_3|^2}{|\lambda_1||\lambda_4||\lambda_5|}\Bigr)\pm\sqrt{3}\left(\theta_1+\theta_4-\theta_5
 +M_1+M_4-M_5\right)&\geq&0,\\ \nonumber
\log\Bigl(\frac{|\lambda_4|^2}{|\lambda_1||\lambda_3||\lambda_5|}\Bigr)\pm\sqrt{3}\left(\theta_1-\theta_3+\theta_5
+M_1-M_3+M_5\right)&\geq&0,\\
\log\Bigl(\frac{|\lambda_5|^2}{|\lambda_1||\lambda_3||\lambda_4|}\Bigr)\pm\sqrt{3}\left(\theta_1+\theta_3-\theta_4
+M_1+M_3-M_4\right)&\geq&0.
\end{eqnarray}
Here  $\theta_i$ denotes the phase of $\lambda_i$ and $M_i\in\mathbb{Z}$
is an integer.
Since the arguments of the logarithm have to be positive
the second line implies $|\lambda_3|\ge \sqrt{|\lambda_1| |\lambda_4||\lambda_5|}$.
Substituting this to an analogous equation
 from the first line, $|\lambda_1|^2\ge \sqrt{|\lambda_3| |\lambda_4||\lambda_5|}$,
we obtain an inequality,  $|\lambda_1|^{3/2}\ge  |\lambda_4|^{3/2} |\lambda_5|^{3/2}$.
Taking both sides to the power $2/3$ and repeating these
steps with the other equations,
we arrive  at the following necessary condition on the eigenvalues:
\begin{equation}\label{necessary}
  |\lambda_\alpha|\geq|\lambda_\beta||\lambda_\gamma|,
\end{equation}
for any choice of indices but the leading one,
$\alpha \ne \beta \ne \gamma \ne 0$. These  conditions are analogous to
 relations (\ref{lambda3}) defining the set ${\cal A}_2^Q$
 of  qubit channels accessible by a semigroup \cite{DZP18,PRZ19},
  except that for $N=3$ these conditions are necessary but not sufficient.

  \subsection{Product probability vectors}
 Applying Eq. \eqref{mprime},
 we obtain eigenvalues of a dynamical semigroup as an explicit function of time:
\begin{eqnarray}
\nonumber\lambda_0&=&1,\\ \nonumber
\lambda_{1,2}&=&\e^{-\frac{3}{2}(t_3+t_4+t_5+t_6+t_7+t_8)}\e^{\pm i\frac{\sqrt{3}}{2}(-t_3-t_4-t_5+t_6+t_7+t_8)},\\ \nonumber
\lambda_{3,6}&=&\e^{-\frac{3}{2}(t_1+t_2+t_4+t_5+t_7+t_8)}\e^{\pm i\frac{\sqrt{3}}{2}(t_1-t_2+t_4-t_5+t_7-t_8)},\\ \nonumber
\lambda_{4,8}&=&\e^{-\frac{3}{2}(t_1+t_2+t_3+t_5+t_6+t_7)}\e^{\pm i\frac{\sqrt{3}}{2}(t_1-t_2-t_3+t_5+t_6-t_7)},\\ \nonumber
\lambda_{5,7}&=&\e^{-\frac{3}{2}(t_1+t_2+t_3+t_4+t_6+t_8)}\e^{\pm i\frac{\sqrt{3}}{2}(t_1-t_2+t_3-t_4-t_6+t_8)},\\ \nonumber
\end{eqnarray}
For channels $\Phi_{\vec p}$ accessible  by a semigroup, the probabilities
can be expressed by Eq.~\eqref{ptime} as explicit functions of time.
These functions are analogous to Eq. \eqref{qubitprob}
valid for $N=2$,
but contain many more terms,
For instance, in the case $N=3$ the time dependence of $p_{0}$
has the form
\begin{eqnarray}
\nonumber p_{0}=\frac{1}{9}\big(1&+&2\e^{-\frac{3}{2}
(t_1+t_2+t_3+t_5+t_6+t_7)}\cos{[\textstyle\frac{\sqrt{3}}{2}(t_1-t_2-t_3+t_5+t_6-t_7)]}\\
  \nonumber&+&2\e^{-\textstyle\frac{3}{2}
(t_3+t_4+t_5+t_6+t_7+t_8)}\cos{[\textstyle\frac{\sqrt{3}}{2}(t_3+t_4+t_5-t_6-t_7-t_8)]}\\
  \nonumber&+&2\e^{-\textstyle\frac{3}{2}
(t_1+t_2+t_4+t_5+t_7+t_8)}\cos{[\textstyle\frac{\sqrt{3}}{2}(t_1-t_2+t_4-t_5+t_7+t_8)]}\\
  \nonumber&+&2\e^{-\frac{3}{2}
(t_1+t_2+t_3+t_4+t_6+t_8)}\cos{[\textstyle\frac{\sqrt{3}}{2}(t_1-t_2+t_3-t_4-t_6+t_8)]}\big).
\end{eqnarray}

\medskip

A particular result characteristic to the single qubit case is the structure of the boundary
$\partial{\cal A}_2^Q$
of the set of accessible maps, which consists of the channels corresponding
to probability vectors with a product structure  \cite{PRZ19},
\begin{equation}\label{qubitproductprob}
  p_0p_i=p_jp_k,
\end{equation}
for any different choice of $i,j,k$ from the set $\{1,2,3\}$.
In the case $N=3$ we are dealing with 8-dimensional simplex and
product probability vectors form at most a 4-dimensional
subset of the 7-dimensional boundary of $\mathcal{A}_3$.
Therefore, there are other channels at the boundary
represented by a probability vector $p$ without
the product structure.

\begin{proposition}\label{pro:classicalproduct}
Consider any four out of eight Lindblad generators $\mathcal{L}_\mu$ defined
 by Eq. \eqref{lmu} whose corresponding unitaries $U_\mu$ are Hermitian
 conjugates up to a phase. Let us represent them by
 $\mathcal{L}_a,\mathcal{L}_b,\mathcal{L}_\alpha,\mathcal{L}_\beta$
  for which, up to a phase, $U_a$ equals to $U_b^\dagger$ and $U_\alpha$
  equals to $U_\beta^\dagger$. For this set the associated dynamical
  semigroup defined by Eq. \eqref{lt} belongs to the 4-dimensional boundaries
   of $\mathcal{A}_3^Q$ and the corresponding  probability vector possesses
   the product structure.
\end{proposition}
\begin{proof}
The corresponding 8--dimensional vector $t_{\mu}$ of interaction times has
four components  equal to zero and four non-zero components,
so the channel  belongs to the boundary  of $\mathcal{A}_3^Q$.
 To see the product structure of the probability vector $p_{\mu}$
  one can express probabilities as  explicit functions of time through Eq. \eqref{ptime}.
 For the probabilities appearing in $\Lambda_{t_{a}}\Lambda_{t_{b}}\Lambda_{t_\alpha}
\Lambda_{t_{\beta}}=\sum p_\mu U_\mu \otimes \overline{U}_\mu$,
one has $\vec{\tilde{p}}=\vec{w}(a,b)\times\vec{w}(\alpha,\beta)$
in which $\vec{\tilde{p}}$ is a permutation of the probability  vector $\vec{p}$,
while $\vec{w}(x,y)=\big(f_1(x,y),f_2(x,y),f_3(x,y)\big)$ is a
local probability vector of size three given by
\begin{eqnarray}
\nonumber
f_1(x,y)&=&\frac{1}{3}\big(1+2\e^{-\frac{3}{2}(t_x+t_y)}\cos[\textstyle\frac{\sqrt{3}}{2}(t_x-t_y)]\big),\\
\nonumber
f_2(x,y)&=&\frac{1}{3}\big(1-\e^{-\frac{3}{2}(t_x+t_y)}\cos[\textstyle\frac{\sqrt{3}}{2}(t_x-t_y)]+\sqrt{3}\e^{-\frac{3}{2}(t_x+t_y)}\sin[\textstyle\frac{\sqrt{3}}{2}(t_x-t_y)]\big),\\
f_3(x,y)&=&\frac{1}{3}\big(1-\e^{-\frac{3}{2}(t_x+t_y)}\cos[\textstyle\frac{\sqrt{3}}{2}(t_x-t_y)]-\sqrt{3}\e^{-\frac{3}{2}(t_x+t_y)}\sin[\textstyle\frac{\sqrt{3}}{2}(t_x-t_y)]\big).
\end{eqnarray}
\end{proof}
A particular case of~\pref{pro:classicalproduct} corresponds
to the channels with real eigenvalues $\lambda_{\mu}$.
Note that this case is not generic: sampling uniformly
vectors $\vec r$, $\vec s\in\Delta_3$ generates
channels with complex eigenvalues with probability one.
This can be seen by decomposing the eigenvalues
of $\Phi_{\vec{p}}$ into its real and imaginary
parts. Equating the imaginary part of each eigenvalue to $0$ impose
a large list of restrictions on $\vec r$ and $\vec s$,
leaving-at most-one free parameter of the original four ($r_1$,
$r_2$, $s_1$, $s_2$). Thus, the region
in the space of mixed unitary channels,
for which all the eigenvalues $\lambda_{\mu}$ are positive has zero measure.
\par
An interesting implication of~\pref{pro:classicalproduct} is that
 inside the set $\mathcal{A}_3^Q$ of accessible maps there are no channels with rank $2$.
 In other words, the edges connecting identity map $U_0$ with the Weyl unitary rotations $U_{\mu}$
contain no accessible maps.
This fact is visualized in Fig.~\ref{fig:4panel}
which presents several cross-sections of the 8D
simplex $\Delta_8$ of mixed unitary Weyl channels and
its subset $\mathcal{A}_3^Q$  of accessible maps.

Another difference  concerning the structure of the set of accessible maps
for $N=2$ and $N=3$ is that in the former case  $p_0$ is the largest
component for the probability vector of $\Phi\in\mathcal{A}_2^Q$
for any $t\geq0$ \cite{PRZ19}.
This is not the case for $N=3$, even if we choose only one
interaction time $t_{\mu}$ to be nonzero -- see Fig.~\ref{fig:4panel}a.
\begin{figure}
  \includegraphics{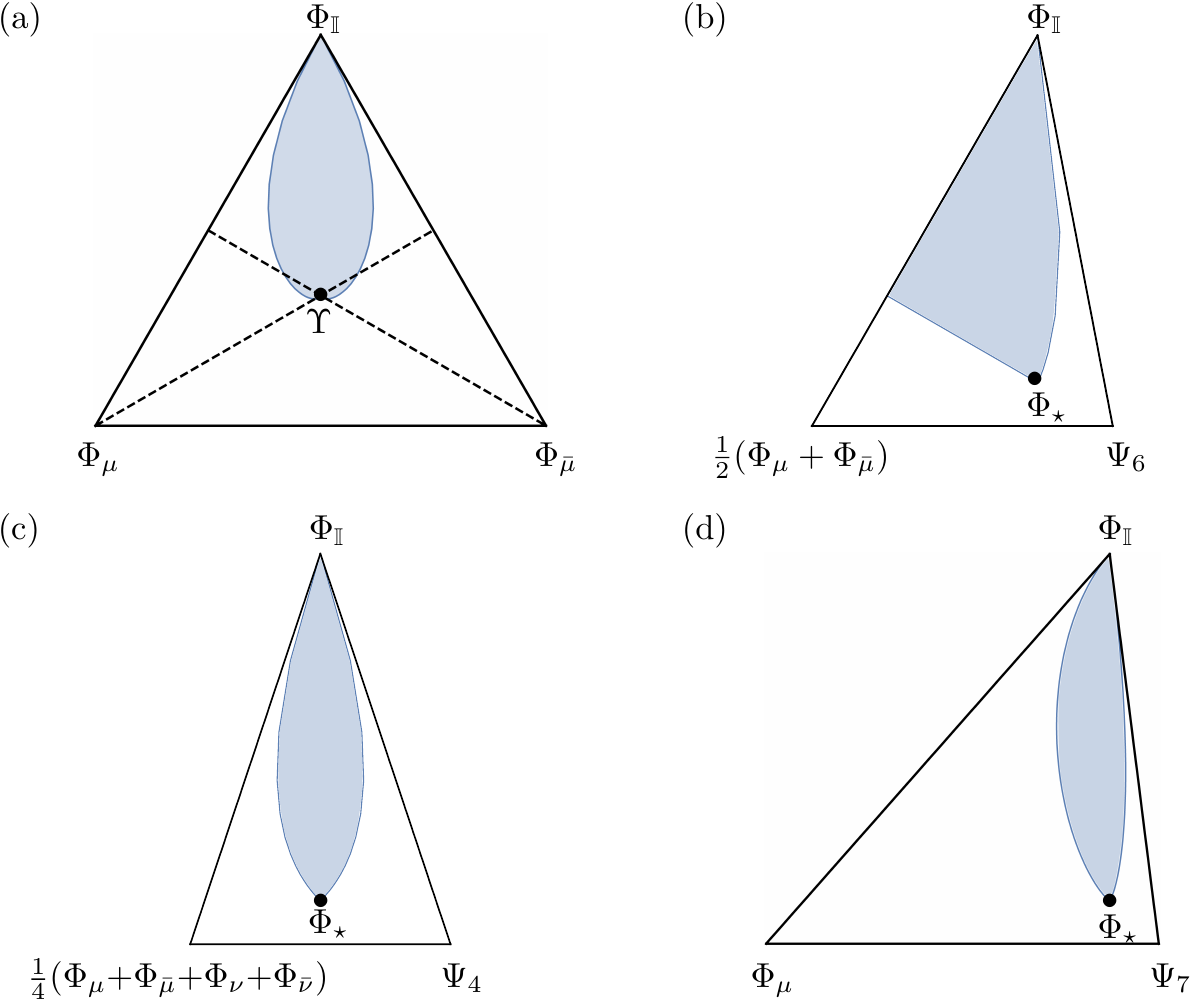}
\caption{
Faces or cross-section of the simplex  $\Delta_8$  of mixed unitary channels
with cross-sections of the set ${\cal A}_3^Q$ of maps accessible by a semigroup
denoted in blue (gray).
Completely depolarizing channel is represented by $\Phi_\star$,
while identity channel reads $\Phi_\Id=\Id\otimes\Id$.
Sections are determined by  $\Phi_{\Id}$
and
 \textbf{a)}
$\Phi_\mu=U_\mu\otimes\overline{U}_\mu$ and
$\Phi_{\overline{\mu}} = U^{\dagger}_\mu\otimes U^{T}_\mu$.
The center of the triangle represents $\Upsilon =
\frac{1}{3}(\Phi_\Id+\Phi_\mu+\Phi_{\overline{\mu}})$;
 \textbf{b)}
$(\Phi_\mu+\Phi_{\overline\mu})/2$, and
$\Psi_6=\frac{1}{6}\sum_{\beta}
\Phi_{\beta}$, with
$\beta\neq 0$, $\mu$, and $\overline\mu$;
 \textbf{c)}
$\frac{1}{4} (\Phi_{\mu}+\Phi_{\overline\mu}
                             + \Phi_{\nu}+\Phi_{\overline\nu})$
and
                             $\Psi_4=\frac{1}{4} (\sum_{\beta}\Phi_{\beta})$
                       with $\beta\ne 0, \mu, {\overline\mu},\nu, {\overline\nu}$;
and
 \textbf{d)}
$\Phi_\mu$, and
$\Psi_7=\frac{1}{7}\sum_{\beta}
   \Phi_{\beta}$, where
$\beta\neq 0$, and $\mu$.
Notice that cross-section of the set ${\cal A}_3^Q$ presented in panel a)
is up to rotation equivalent to the support of the spectra of $N=3$ accessible channels
    shown in Fig. \ref{fig:accesible}a.
}
\label{fig:4panel}
\end{figure}

 \subsection{Accessible maps for qubits and qutrits}
  To describe properties of the set  $\cal{A}_N$ of maps accessible by a semigroup
    we shall compare the volumes of these sets for $N=2$ and $N=3$
    relative to the volumes of the simplex $\Delta_{N^2-1}$ of Weyl mixed unitary channels.
   In the qubit case the set of mixed unitaries corresponds to the tetrahedron
   spanned by the identity operation, $\Phi_0=\Id\otimes\Id$,
       and unitary rotations associated to three Pauli matrices,
$\Phi_1=\sigma_x\otimes\overline{\sigma}_x$,
$\Phi_2=\sigma_y\otimes\overline{\sigma}_y$, and
$\Phi_3=\sigma_z\otimes\overline{\sigma}_z$.
As the length of each edge of the regular tetrahedron can be set to unity,
its  volume reads, $V_{\Delta_3}=\sqrt{2}/12$.

To calculate the volume of its subset ${\cal A}_2^Q$ we use
Eq.~\eqref{qubitprob}
and work with three-dimensional Cartesian coordinates.
Each point inside the tetrahedron $\Delta_3$ can specified by its coordinates
$X,Y,Z$.
 For channels in $\Delta_3$ which also belong to ${\cal A}_2^Q$, the
 probabilities $p_0$, $p_1$, $p_2$, and $p_3$ are defined by Eq.~\eqref{qubitprob}
 suggesting $X$, $Y$, and $Z$ as an explicit function of $t_1$, $t_2$, and $t_3$.
Let $|J|=\exp[-4(t_1+t_2+t_3)]/\sqrt{2}$ denotes the Jacobian determinant associated to such a change of variables.
The volume of the set  of one-qubit accessible maps reads then
\begin{eqnarray}
V_{\c A_2}=\iiint_{\c A_2}\mathrm{d}X\mathrm{d}Y\mathrm{d}Z=\iiint|J|\mathrm{d}t_1\mathrm{d}t_2\mathrm{d}t_3=
\frac{1}{\sqrt{2}}
\int_0^\infty\e^{-4t_1}\mathrm{d}t_1\int_0^\infty
\e^{-4t_2}\mathrm{d}t_2\int_0^\infty\e^{-4t_3}
\mathrm{d}t_3=\frac{1}{4^3\sqrt{2}}.
\end{eqnarray}
Hence the relative volume of the set of maps accessible by a semigroup is
${V_{\c A_2}}/{V_{\Delta_3}}=3/32=0.09375$,
as it forms a quarter of the volume of the set
${\c U_2^Q}$ of one-qubit unistochastic channels
investigated in \cite{NA07,Si19,Si20}.
On the other hand,  for the qutrit channels
the analogous ratio gained numerically is much smaller,
${V_{\lindblad}}/{V_{\Delta_8}}= 9.02\times10^{-4}\pm
1.8\times 10^{-5}$. This estimation was obtained by sampling
$10^{6}$ different random probability vectors $\vec{p}\in \Delta_8$
and verifying whether the corresponding mixed unitary channel
satisfies all conditions  (\ref{time}) for accessibility.
As for qubit channels this ratio  is larger by two orders of magnitude, one
can expect a significant decrease of the relative volume of accessible channels
with their dimension $N$.
\medskip

To analyze the structure of the set $\c A^Q_3$ of accessible maps
we shall analyze cross-sections of the simplex   ${\Delta_8}$
determined by fixing the weight $p_0$ of the identity component
of a mixed unitary channel (\ref{mixeduni}).
  Consider an $8$-point random vector
   $\vec{\varphi} \in {\Delta_7}$ from the $7$ dimensional simplex
 and define a $9$ dimensional vector,
 $\vec{p} = [p_0, (1-p_0)\vec{\varphi}]$.
The corresponding mixed unitary channel $\Phi_{\vec{p}}$ is  thus given by Eq. (\ref{phi}).

Generating auxiliary vectors  $\vec{\varphi}$ according to the flat measure in the simplex ${\Delta_7}$
we obtain an ensemble of random vectors $\vec{p}$
and thus random unitary channels with a fixed size of the identity component $p_0$ and
 study the fraction of probability vectors corresponding to accessible channels.
 Numerical results presented in ~Fig.~\ref{testd}
 show that  for  $N=3$  the relative volume of the analyzed set $\lindblad$,
grows monotonically to $1$ as a function of the
weight $p_0$ of the probability vector defining the channel.
For comparison we plotted analogous data obtain for the case $N=2$
for which the relative volume grows faster with the weight $p_0$.
However, in this case the curve has an inflection point
which does not occur for $N=3$. The inset shows behavior for small values of $p_0$,
as the relative volume is positive already for $p_0 <1/9$.
This observation is consistent with Fig.~\ref{fig:4panel}a,
showing that in the case $N=3$
 there exist accessible maps $\Phi_p$,
  for which the component $p_0$ is not the largest one.
\par
\begin{figure}[ht]
    \centering
    \includegraphics[]{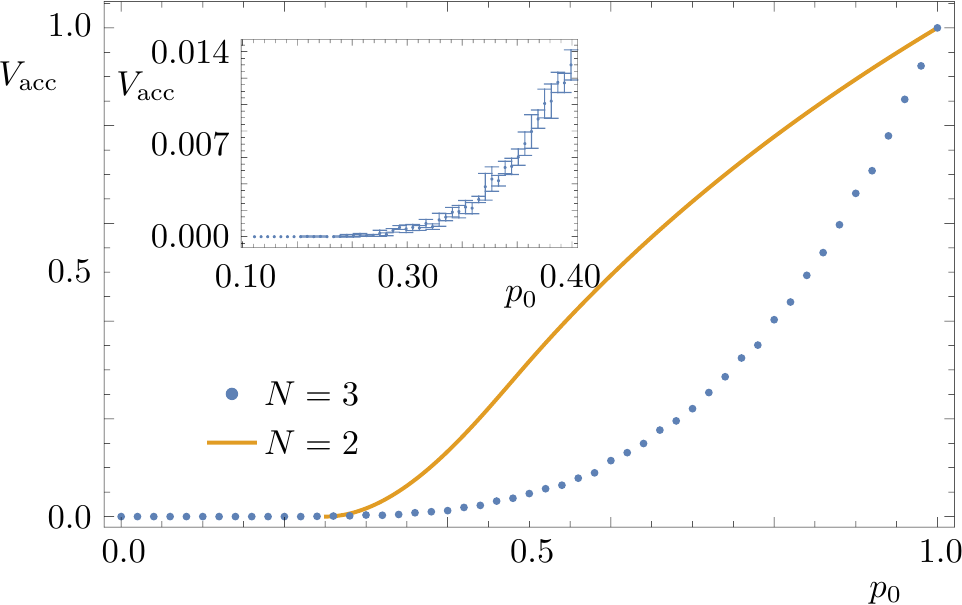}
    \caption{
    Relative volume $V_{\mathrm{acc}}$ of the set ${\cal A}_3^Q$ of accessible
channels at a cross-section of the probability
simplex $\Delta_8$ of mixed unitary channels specified by the fixed value of the identity component $p_0$.
      Dotted blue line represents numerical results   for $N=3$ estimated by the fraction of accessible maps.
      Each point was calculated using samples of size $10^5$ channels, so
      the error bars are smaller than the size of the dot.
      Inset zooms in  the part of the plot, for which  the relative volume is positive.
    Continuous orange curve
      represents analytical results obtained for $N=2$. }
    \label{testd}
\end{figure}

 \subsection{Bistochastic matrices accessible by $N=3$  classical semigroup}

As discussed in  Section \ref{sec:six},
any mixed Weyl channel decoheres to a classical circulant bistochastic matrix.
We shall analyze the issue here  for $N=3$ in some more detail.
Related question of describing the subset of the Birkhoff polytope
containing these bistochastic matrices of order three,
so that their square root is bistochastic, was studied in \cite{SZ18}.

Consider a random Weyl channel $\Phi_{\vec p}$
distributed uniformly from the triangle  $\Delta(\Phi_\Id,\Phi_{X}, \Phi_{X^2})$
shown in  Fig. \ref{fig:4panel}a
-- a 2D face of the $8D$ simplex $\Delta_8$ of all Weyl channels.
Circulant bistochastic transition matrices, obtained by hyper-decoherence,
$T  = \c D_{h}(\Phi_{\vec p})$ cover the triangle
$\Delta(\Id_3, X_3, X_3^2)$ uniformly -- see Fig.  \ref{DAfig1}a.
This is a consequence of the fact that in this  subspace the supermap
 $\c D_{h}$ is linear and transforms a regular triangle into itself,
\begin{equation}
 \c D_{h}\bigl(\sum_{j=0}^2 q_j \; \Phi_{X^j}\bigr) =
  \sum_{j=0}^2 q_j  X^j  = T_q.
\label{trans_triangle}
\end{equation}
 On the other hand, if one takes an average over the entire set
 ${\cal B}_3^Q$ of Weyl channels and samples the probability
 vector $\vec p$ uniformly in the  simplex $\Delta_8$
 the distribution of classical matrices exhibits maximum at the
 flat matrix $T_*$ with all entries equal to $1/3$
 --  see Fig. \ref{DAfig1}b.
 If the sampling is restricted to the set of accessible quantum channels only,
 its image under decoherence is supported inside the characteristic
 `heart--like' shape presented in panel  Fig.~\ref{DAfig1}c.

\begin{figure}[h]
  \centering
  \includegraphics[]{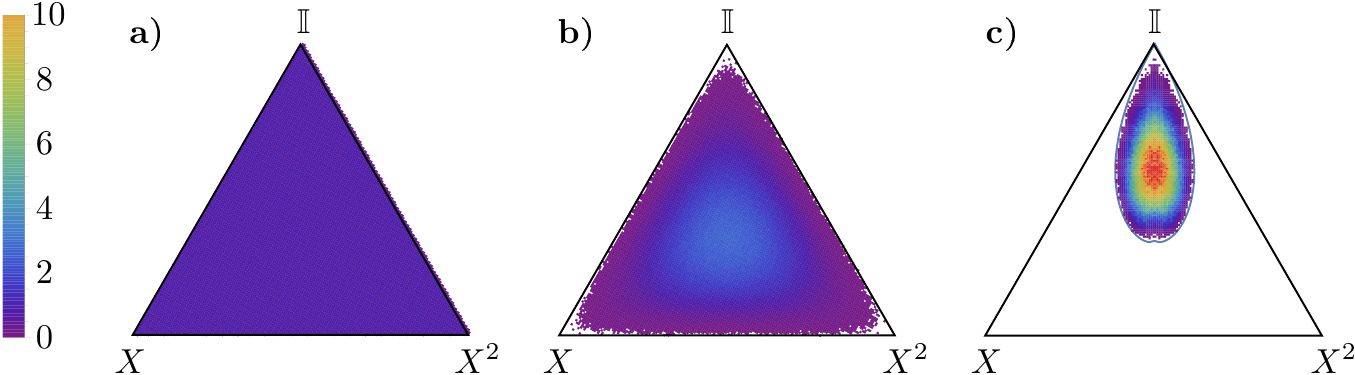}
  \caption{ Distribution of random circulant
  bistochastic matrices of size $N=3$ inside the regular triangle
  $\Delta(\Id_3, X_3, X_3^2)$ obtained by hyper-decoherence of:
 \textbf{a)}  random Weyl channels distributed uniformly from the 2D triangle
  $\Delta(\Phi_\Id,\Phi_{X}, \Phi_{X^2})$;
  \textbf{b)}
   random channels distributed uniformly from the entire
  8D simplex of Weyl channels;
  \textbf{c)}    random channels distributed uniformly from the
     set ${\cal A}_3^Q$ of quantum accessible channels.
      Light blue line represents the boundary of the set
       ${\cal A}_3^C$ of matrices accessible by classical semigroups --
      observe similarity with  Fig. \ref{fig:accesible}a.
       }
\label{DAfig1}
\end{figure}

In Eq. (\ref{fourcirc}) we used the fact that a circulant matrix $T$
can be diagonalized by the Fourier matrix $F$.
If the transition matrix is obtained by the convex combination of permutation matrices,
$T=\sum_{i=0}^2 q_j X_3^j$, then its spectrum consists of $1, z, {\bar z}$,
where $z=\sum_{i=0}^2 q_j \omega_3^j$.
Therefore, if the distribution of the probability vector $q$ is uniform in the
triangle of bistochastic matrices, so is the level density of subleading eigenvalues
in the regular triangle  $\Delta(1, \omega_3, \omega_3^2)$
inscribed in the unit disk ---compare Fig.~\ref{DAfig3}a.

To analyze the shape of the set ${\cal A}_3^C$ of  circulant
  bistochastic matrices  accessible by a classical semigroup
 we may use Proposition \ref{prop14}.
 Since the hyper-decoherence transformation, $T  = \c D_{h}(\Phi_{\vec p})$,
  acts linearly  (\ref{trans_triangle})
  in the
  triangle of quantum Weyl channels  $\Delta(\Phi_\Id,\Phi_{X}, \Phi_{X^2})$,
  the shape of the set  ${\cal A}_3^C$ of classically accessible matrices
  is identical with cross-section of the set
   ${\cal A}_3^Q$ of quantum accessible maps
   presented in Fig. \ref{fig:4panel}d.
   This set is uniformly covered when the same region,
   corresponding to quantum channels, is uniformly sampled.
   The structure of this star-shaped set, bounded by fragments of two
   symmetric logarithmic spirals, is discussed in Section \ref{AppSpiral}.
   This set appeared in the literature in context of
    analyzing Markovian evolution of quantum coherence
    and $N=3$ embeddable stochastic matrices \cite{LKM17}.

\medskip
  \begin{figure}[h]
 \centering
 \includegraphics[]{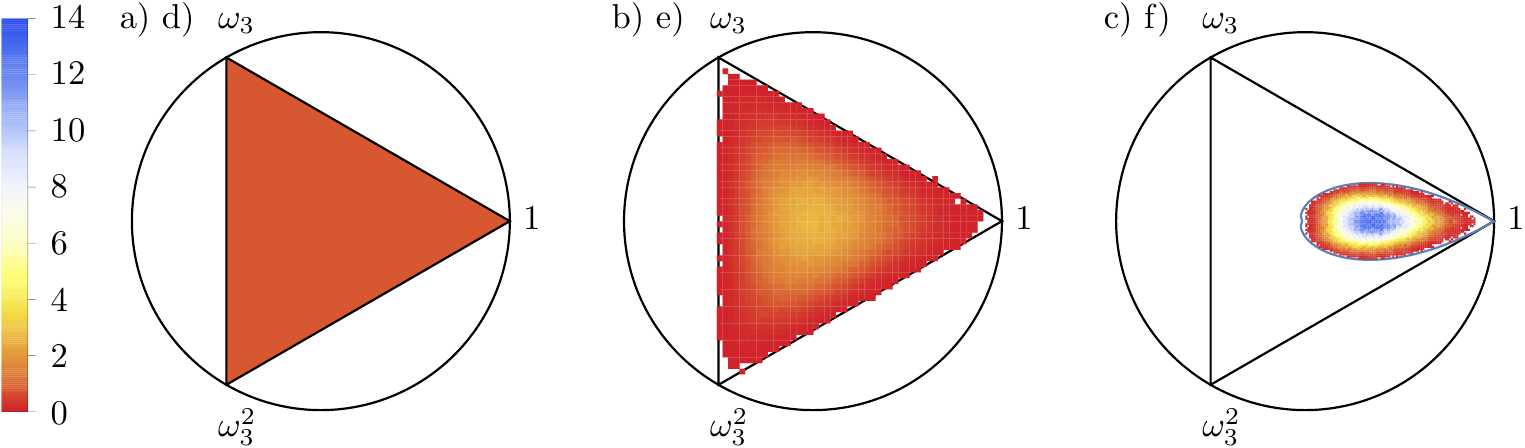}
  \caption{Probability distribution for
   superimposed spectra of $10^{4}$
      quantum maps acting on $3$-level systems:
       \textbf{a)} Weyl channels uniformly distributed in the triangle  $\Delta(\Phi_\Id,\Phi_{X}, \Phi_{X^2})$,
      \textbf{b)} random Weyl channels $\Phi_{\vec p}$ choosing  $\vec p$ from the flat distribution on $\Delta_8$,
     \textbf{c)} random Weyl channels distributed uniformly from the set $\mathcal{A}^{Q}_3$ of accessible quantum channels and spectra of the corresponding random circulant bistochastic matrices of size $N=3$:
       \textbf{d)} flat measure on the simplex  $\Delta_2$ of circulant bistochastic matrices,
      \textbf{e)} obtained by hyper-decoherence $T(\Phi_{\vec p})$ of random Weyl channels $\Phi_{\vec p}$ from $\Delta_8$,
     \textbf{f)} accessible circulant bistochastic matrices  $\mathcal{A}^{C}_3$, obtained  by hyper-decoherence $T(\Phi_{\vec p})$ of accessible quantum channels  from $\mathcal{A}^{Q}_3$ --
         the boundaries of the support are shown in  Fig. \ref{fig:accesible}a.
       Due to Proposition \ref{prop-channel-hyperdeco}
    panels d), e) and f) coincide with panels a), b), and c) respectively,
       so they are not plotted.}
\label{DAfig3}
\end{figure}

\section{Set  $\unicha$ of unistochastic channels acting on $N=3$ systems}
\label{sec:five}

As for $N=2$ the set ${\cal A}_2^Q$ of quantum accessible channels
forms a quarter of the set  ${\cal U}_2^Q$ of unistochastic channels,
so  the relation
${\cal A}_2^Q  \subset {\cal U}_2^Q$ holds \cite{PRZ19}.
Hence in this Section we  compare the corresponding sets for $N=3$.
A similar analysis
will be performed for the analogous classical sets inside the Birkhoff polytope
${\cal B}_3^C$ of bistochastic matrices of this size.

\subsection{Unistochastic channels and channels accessible by a semigroup} 

We start demonstrating that for one-qutrit system
the set  of quantum accessible channels is not included
in the set  of unistochastic channels,
written ${\cal A}_3^Q  \nsubseteq {\cal U}_3^Q$.

\begin{proposition}
 For $3$-level systems,
there exists a channel accessible by a Lindblad semigroup
which is not unistochastic.
\end{proposition}
\begin{proof}
  To show a  counterexample we analyze Weyl channels
  being a convex combination of $\Id$, $Z$, and $Z^2$,
  \begin{equation}
    \Phi_p = p_1 \Id\otimes\Id + p_2 Z\otimes \overline{Z} +
    (1-p_1-p_2) Z^2\otimes\overline{Z^2}.
    \label{eq:crossSectionZZ2}
  \end{equation}
  We wish to find a vector $p=(p_1,p_2,p_3)$
  such that there is no unitary dilation matrix $U\in U(9)$
  for which Eq. (\ref{eq:defuni}) produces the analyzed channel $\Phi_p$.
  Given that $Z$ is diagonal, the dilation matrix $U$ is block diagonal,
  with three unitary blocks $W_1,W_2$ and $W_3$, each of size $3$,
  \[
    U = W_1\oplus W_2 \oplus W_3
    =
  \diag(1,0,0)\otimes W_1 +\diag(0,1,0)\otimes W_2
+\diag(0,0,1)\otimes W_3.
  \]
  By unistochasticity of the channel $\Phi_p$ the
  corresponding Choi matrix $D=\Phi_p^R$ has to satisfy
  $D=\frac{1}{N}(U^R)^{\dagger}U^R$.
  This leads to the following conditions for $W_1$, $W_2$ 
  and $W_3$:
\begin{equation}
  \tr\ W_2^{\dagger}W_1 =
  s,\quad
  \tr\ W_3^{\dagger}W_1 =
s^{*},\quad
  \tr\ W_3^{\dagger}W_2 =
  s,
  \label{eq:unistochasticTest}
\end{equation}
where $s =  \frac{1}{2}\bigl(3p_1-1\bigr) +i\frac{\sqrt{3}}{2}\bigl(p_1+2p_2-1\bigr)$.

\par
Consider a probability vector with three non zero entries
given by $\vec{p}_c = (p_1,p_2,1-p_1-p_2) = (35,4,1)/40$.
Using $\vec p_c$ in Eq.~\eqref{eq:crossSectionZZ2},
we can verify that $\Phi_{\vec{p}_c}$ satisfies the three properties mentioned at
the end of Subsection~\ref{lindbladchar},
which implies that the channel is accessible, $\Phi_p \in{\cal A}_3^Q$.
Making use of the parametrization of the group of unitary matrices of size $3$
one can show that there are no
 unitary matrices $W_1, W_2, W_3\in U(3)$
which satisfy conditions  (\ref{eq:unistochasticTest})
for the vector $\vec{p}_c$ selected above:
if any two of these equalities are satisfied, the third one is not.
Hence for this $\vec{p}$ there is no unitary dilation matrix $U\in U(9)$
entering expression  (\ref{eq:defuni}),
so the channel  $\Phi_{\vec{p}}$ is not unistochastic.
\end{proof}

Although for $2$-level systems,
the following inclusion relations holds
${\cal A}_2^Q  \subset {\cal U}_2^Q$~\cite{PRZ19},
we have shown that the analogous property is not valid for $N=3$.
\par

\subsection{Unistochastic and $N$-unistochastic matrices} 

Due to the process of hyper-decoherence a quantum stochastic
map undergos transition to a classical stochastic matrix.
We shall  analyze relations between
subsets of the set ${\cal S}_N$  of stochastic matrices
introduced in Section  \ref{bisto}.

\begin{proposition}
For any $N$-level system, the set of
unistochastic matrices  ${\cal U}_N^C$ is contained in the set
${\cal D}_h({\cal U}_N^Q)$ of transition matrices
obtained by the hyper-decoherence \cite{BZ17} of a unistochastic channel.
\end{proposition}
\begin{proof}
It is enough to take a dilation matrix of the product form,
$U=V\otimes {\1}_N$.
Then due to
 Eq.~\eqref{eq:transitionmatrix}  one obtains $T=V \odot {\bar V}$,
where $\odot$ denotes the Hadamard product,
so the classical transition matrix is unistochastic.
Thus the set of classical transition matrices obtained by hyper-decoherence
of unistochastic channels ${\cal U}_N^Q$
containing $N$-unistochastic matrices,
includes the set  of unistochastic matrices,
${\cal D}_h({\cal U}_N^Q)={\cal U}_{N,N} \supset  {\cal U}^C_N$
-- compare table  \ref{tab:sets}.
\end{proof}

\subsection{Quantum unistochastic maps and classical
unistochastic matrices}
\label{sec:unimapsmatrices}
To explore relations between unistochastic maps and
matrices, we will restrict our attention to the face
of the set ${\cal W}_N$ of mixed Weyl channels determined by three points:
$\Id$, $U_X$, and  $U_{X^2}$.
In this equilateral triangle we need to distinguish two regions:
\par
i) the  region bounded by the 3-hypocycloid \cite{volume} denoted by
$\mathbb{H}_3$. For
$p_1$, $p_2 \in \Delta_2$, the set $\mathbb{H}_3$
includes points satisfying the inequality,
\begin{equation*}
  4p_1^2(1-p_1-p_2)p_2-(p_1-p_1^2-(1-p_1-p_2)p_2)^2 \geq 0.
\end{equation*}
The $3$--hypocycloid
determines the boundary of the set of unistochastic matrices
at the cross-section of the $N=3$ Birkhoff polytope \cite{volume},
 as the above expression implies positivity of the
 quantity (\ref{eq:jarlskog}), which
 determines  the Jarlskog invariant of
the corresponding bistochastic matrix---see Appendix \ref{AppA}.

ii) the region resembling the Star of David, formed
by two equilateral triangles, $\triangle_A$ and
$\triangle_B$, each one described by
three vertices on the plane:
\begin{equation*}
    \triangle_A\colon (1/3,0), (1/3, 1/\sqrt{3}), (5/6,
    1/2\sqrt{3}), \qquad
    \triangle_B\colon  (2/3,1/\sqrt{3}), (2/3,0), (1/6,1/2\sqrt{3}).
\end{equation*}
These subsets of the triangle $\Delta(\Id,\Phi_{X},  \Phi_{X^2})$
 are represented in Fig.~\ref{fig:singlepanel}.
 The $3$--hypocycloid plays a key role \cite{volume}
 in characterizing the set of unistochastic matrices of order three.
 We will show now that it is also relevant to describe
 the set of unistochastic channels and its cross-sections.

\begin{figure}[ht]
  \includegraphics{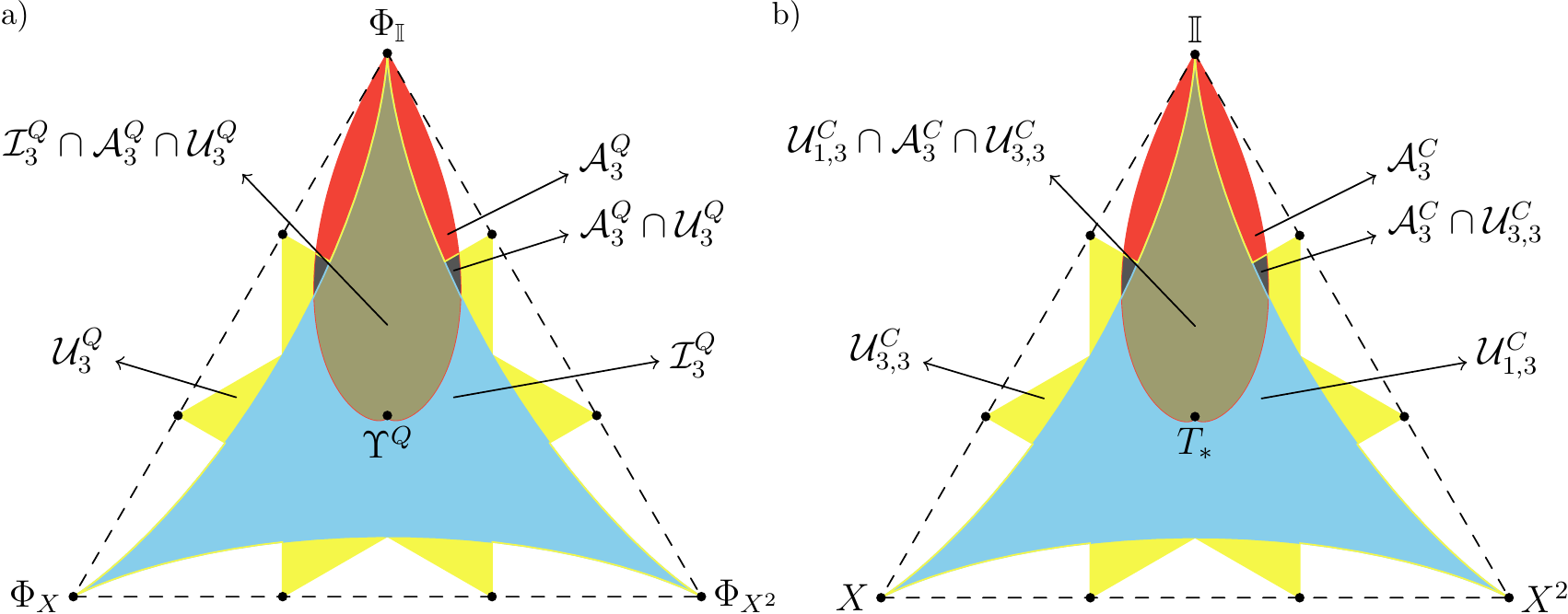}
  \caption{
  a) Face  of  8D the set of $N=3$ Weyl channels  $\Phi_{\vec{p}}$
    along the triangle   $\Delta(\1,\Phi_{X}, \Phi_{X^2})$,
    with the symmetric mixture
$\Upsilon^{Q} =\frac{1}{3}(\Phi_\Id+\Phi_{X}+\Phi_{X^2})$
in its center. Dashed subsets denote
  i) blue hypocycloid $\mathbb{H}_3$ represents
   isometric unitary channels ${\cal I}_3^Q$
   ii)  yellow region forming a David Star  $\triangle_A\cup \triangle_B$
      belongs to the set $\unicha$ of unistochastic channels;
    iii) red heart-shape region
       denotes set $ {\cal A}_3^Q$ of channels accessible by a Lindblad semigroup.
b) Triangle  $\Delta(\Id_3, X,  X^2)$ of $N=3$ circulant bistochastic matrices
   obtained by  hyper-decoherence of  Weyl channels:
  i) blue hypocycloid $\mathbb{H}_3$ represents
       the set ${\cal U}_3^C$ of
    unistochastic transition matrices, $T=V \odot {\bar V}$,
   ii)  yellow region forming a David Star represents
      the set ${\cal U}_{3,3}^C$ of u
     $3$-unistochastic matrices of order $3$,
    iii) red heart-shape region
       denotes set $ {\cal A}_3^C$ of embeddable matrices
        accessible by a classical semigroup.
        Note that $T_* = \mathcal{D}_h(\Upsilon^{Q})$
         represents the flat transition matrix at the center
         of the Birkhoff polytope $\mathcal{B}_3^{C}$.}
    \label{fig:singlepanel}
\end{figure}

\begin{proposition}
Mixed unitary Weyl channel $\Phi_{\vec p}$
belonging the the $3$--hypocycloid,
 $\vec{p}\in\mathbb{H}_3 \subset
 \Delta(\Phi_\Id,\Phi_{X}, \Phi_{X^2})$,
 is unistochastic,  $\Phi_{\vec{p}}\in {\cal U}_3^Q$,
and the corresponding transition matrix is unistochastic,
$T(\Phi_{\vec{p}}) \in {\cal U}_3^C$.
\label{prop1}
\end{proposition}
\begin{proof}
Consider a matrix $U$ of order $N^2=9$,
\begin{equation}
    U =
    \sqrt{p_1}e^{i\vartheta_1}\Id\otimes\Id+
    \sqrt{p_2}e^{i\vartheta_2}X \otimes X +
    \sqrt{p_3}e^{i\vartheta_3}X^2 \otimes X^2,
\end{equation}
where $\{\vartheta_i\}_{i=1}^{3}$ are real phases to be
determined.
Substituting $U$ in Eq.~\eqref{eq:defuni}, we observe that it generates a
superoperator
$\Phi_{\vec{p}} = p_1 \mathbb{I}\otimes\mathbb{I} + p_2 X \otimes\overline{X}
+(1-p_1-p_2)X^2\otimes\overline{X^2}$.
Thus we have constructed a channel of the desired form,
provided matrix $U$ is unitary.
To analyze under what conditions it is the case
we will compute the scalar product between each pair of rows.
Let us introduce auxiliary real parameters,
 $L_1 = \sqrt{1-p_1-p_2}\sqrt{p_1}$,
$L_2 = \sqrt{1-p_1-p_2}\sqrt{p_2}$, and $L_3 = \sqrt{p_1}\sqrt{p_2}$.
Then the  resulting expression takes the form,
\begin{equation}
  \exp{(i \varphi_1)} L_1 +
  \exp{(i \varphi_2)} L_2  +
  \exp{(i \varphi_3)} L_3 = 0,
 \label{eq:orthonormal}
\end{equation}
where $\{\varphi_i\}_{i=1}^{3}$ correspond to linear
combinations of phases $\{\vartheta_i\}$.
Left-hand side of Eq.~\eqref{eq:orthonormal}, can be
pictured as three joined line segments on the complex plane,
of length $L_1$, $L_2$, and $L_3$, respectively.
Inequality holds if these three sides
form a triangle. Thus Eq.~\eqref{eq:orthonormal}
is equivalent to the triangle inequality:
\begin{equation}
|L_1-L_2| \leq L_3 \leq L_1 + L_2,
\label{eq:triangle}
\end{equation}
This inequality is satisfied for any point $p$ in the interior of the hypocycloid,
$\vec{p}\in\mathbb{H}_3$.
In such  a case the columns of $U$ form an orthonormal set,
hence $U$ is unitary ~\cite{volume}.
This implies that if $\vec{p}\in\mathbb{H}_3$
then due to Eq. (\ref{eq:defuni})
unitary $U$ leads to a unistochastic $\Phi_{\vec{p}}$,
while $T(\Phi_{\vec{p}})$ is a unistochastic matrix.
\end{proof}

Further analysis shows that the cross-section of the set  ${\cal U}_3^Q$
of unistochastic channels by the plane containing the triangle
$\Delta(\Phi_\Id,\Phi_X, \Phi_{X^2})$ is larger than the
hypocycloid $\vec{p}\in\mathbb{H}_3$.
Numerical results allow us to formulate the  following conjecture.

\begin{conj}
In the face of the simplex ${\cal W}_3$ of the Weyl channels formed by
$\Delta_Q=\Delta(\Phi_\mathbb{I}$, $\Phi_{X}$, $\Phi_{X^2})$,
the set $\unicha$ of unistochastic channels corresponds
to the region formed by the union $\mathbb{H}_3\cup\triangle_A\cup \triangle_B$.
\end{conj}

The triangle $\Delta_Q$ of quantum Weyl channels
shown in Fig. \ref{fig:singlepanel}a
corresponds to classical triangle of bistochastic matrices
$\Delta_C=\Delta(\mathbb{I}, X, X^2)$,
which forms a cross-section of the Birkhoff polytope  \cite{volume}.
Its subset ${\cal A}_3^C$ of matrices accessible by a classical semigroup
has the same structure as the set
${\cal A}_3^Q$ of accessible quantum maps  shown in Fig. \ref{DAfig1}c.
The region union $\mathbb{H}_3\cup\triangle_A\cup \triangle_B$
corresponds to the set ${\cal U}^C_{3,3}$ of $3$-unistochastic matrices
while $\mathbb{H}_3$ denotes its subset
${\cal U}^C_{3}$ of unistochastic matrices -- see Fig. \ref{fig:singlepanel}b.

To conclude note that not every unitary matrix $U\in U(9)$
 produces by  Eq. (\ref{eq:defuni})
 a unistochastic channel of order three,
 which decoheres to a unistochastic transition matrix.
 An exemplary  channel  $\Phi \in {\cal U}_3^Q$
  such that $T(\Phi) \notin {\cal U}_3^C$
 is described  in Appendix~\ref{app}.

  \section{Concluding remarks}
  \label{Sec7}
  In this work we investigated properties of mixed unitary Weyl  channels
  and circulant bistochastic matrices
  which describe corresponding classical dynamics.
  Extending one-qubit results presented in \cite{PRZ19,DZP18}
  we characterized the set ${\cal A}_N^Q$
  of Weyl channels accessible by a quantum semigroup.
  We showed that this set is log-convex and star-shaped with respect to
  the completely depolarizing channel.

  The standard process of decoherence, occurring due to inevitable
  interaction of the system analyzed with an environment,
  transforms any quantum state $\rho$ into the classical probability vector,
  $p={\cal D}(\rho)={\rm diag}(\rho)$, and destroys all quantum effects.
  In a similar way one analyzes analogous transitions
  between a quantum operation $\Phi$ and
  a classical transition matrix, $T={\cal D}_h(\Phi)$.
  Such an effect, sometimes called hyper-decoherence,
  can also act on a Lindblad operator $\cal L$,
  which generates  a quantum semigroup.
  The image gives a legitimate Kolmogorov operator,
   ${\cal K}={\cal D}_h({\cal L})$,
  related to a classical semigroup \cite{CMMV13}.

  We showed that any Weyl channel $\Phi_{\vec p}$
  subjected to hyper-decoherence yields a circulant
   transition bistochastic matrix $T_{\vec q}$, where
   the $N$-point probability vector  $\vec q$
   arises as the marginal of the $N^2$-point initial probability vector
   $\vec p$. Furthermore, for mixed Weyl channels the dynamics
   commutes with decoherence, so that
   $ \c D_{h}\left(\e^{t\c L}\right)=\e^{t\c D_{h}\left(\c L\right)}$.
   Hence   the set ${\cal A}_N^C$
   of transition matrices accessible by a classical semigroup,
   arises by hyper-decoherence of its quantum counterpart,
   ${\cal A}_N^C= {\cal D}_h\bigl( {\cal A}_N^Q\bigr)$.
   This implies that this classical set inherits key properties of the
   its quantum relative  ${\cal A}_N^Q$ and is also log-convex and star-shaped.

   For concreteness, we set the dimension $N=3$
   and in this case investigated geometry of the above sets.
   In particular, we studied various cross-sections
   of $8D$ set ${\cal A}_3^Q$ of accessible quantum channels
   and described in details the corresponding set
   ${\cal A}_3^C$ of  embeddable bistochastic matrices \cite{LKM17},
   which are accessible by a classical semigroup.
  Motivated by the relation between maps accessible by a quantum semigroup
  and  unistochastic channels established in the single qubit case,
  ${\cal A}_2^Q  \subset {\cal U}_2^Q$~\cite{PRZ19},
  we studied the unistochastic operations.
  A constructive example of a qutrit channel is presented,
  which is accessible by a quantum semigroup
  but it cannot be obtained by the partial trace over the environment of the
  same size, initially prepared in the maximally mixed state.
  Such a counterexample  allows us to conclude that
  for higher dimensions such a relation
  between unistochastic and accessible channels breaks down,
  as the analogous inclusion relation does not hold already for $N= 3$.

A notion of higher order, $k$--unistochastic matrices of size $N$ has been introduced.
We have shown that unistochastic channels acting on $N$-level systems
 decohere to $N$--unistochastic matrices,
   ${\cal D}_h({\cal U}_N^Q)={\cal U}^C_{N,N}$.
  It contains the set of unistochastic matrices,
  which arise by hyper-decoherence of isometric unitary transformations,
   ${\cal U}^C_{N}={\cal D}_h({\cal I}_N^Q)$.

  Let us conclude the paper presenting a list of some open problems.
  \begin{enumerate}
    \item Results of this work are obtained for mixed unitary channels
      written in the Weyl basis (\ref{weyl1}).
        Analyze which of them can be generalized for a) mixed unitary channels
          written in other orthogonal basis, b) any bistochastic maps, or
          c) any stochastic maps.
    \item For one-qubit case, channels belonging to Pauli
      semigroups are the only ones connected to the identity map
      through a trajectory formed by mixed unitary channels.
        Is an analogous property true for higher dimensions?
    \item One-qubit Pauli maps belonging to $\mathcal{A}_2^Q$
       (accessible by a semigroup) are of rank $1,2$ or $4$.
       For qutrits, we identified only  accessible channels of rank $1$, $3$,
      and $9$, so existence of $N=3$ accessible channels with rank $2,4,\dots,8$
        remains open.
    \item Check if  the rank of the operation
        does not change (or does not decrease)
          during the time evolution  induced by a quantum
           semigroup $\cal L$ specified in Eq. (\ref{lt}).
    \item
        Find necessary or sufficient criteria for a $N=3$ bistochastic  map $\Phi$
        to be unistochastic,
         so there exists a unitary matrix $U\in U(9)$
          leading to the representation \eqref{eq:defuni}.
  \end{enumerate}

\acknowledgements
It is a pleasure to thank Fabio Benatti, Adam Burchardt,
Dariusz Chu{\'s}ci{\'n}ski, David Davalos,
Kamil Korzekwa, Carlos Pineda, {\L}ukasz Rudnicki
and Konrad Szyma{\'n}ski
for numerous discussions and helpful correspondence.
This research was support by  National Science Center in Poland
under the Maestro grant number DEC-2015/18/A/ST2/00274
and by Foundation for Polish Science under the grant Team-Net NTQC
number  17C1/18-00. Z.P. acknowledges the support by the Polish National 
Science Center under the Project Number 2016/22/E/ST6/00062.
DAA acknowledges support from projects CONACyT 285754,
UNAM-PAPIIT IG100518.
\appendix
\label{app}
\section{Jarlskog invariant for $N=3$ bistochastic matrices} 
\label{AppA}

Consider the Birkhoff polytope ${\cal B}_3^C$
of  bistochastic matrix of order three.
Any matrix from this four-dimensional set
can be parameterized by
four non negative numbers $b_1,b_2,b_3,b_4$,
\begin{equation}
B =
    \begin{pmatrix}
    b_1 & b_2 & 1-b_1-b_2 \\
    b_3 & b_4 & 1-b_3-b_4 \\
    1-b_1-b_3 & 1-b_2-b_4 & b_1+b_2+b_3+b_4-1 \\
    \end{pmatrix}.
\end{equation}
Note that these parameters do satisfy some constraints,
as any entry of $B$ has to be positive and not larger than one.
A bistochastic matrix $B$ is called unistochastic
if  there  exists a unitary matrix $V$
such that $B$ can be represented by the Hadamard (entrywise) product,
$B=V \odot {\bar V}$,
so its entries read  $B_{ij} = \vert V_{ij}\vert^2$.

For any matrix $B \in {\cal B}_3^C$  one introduces the quantity
\begin{equation}
\label{eq:jarlskog}
{\mathcal Q}(B) \equiv
    4b_1b_2b_3b_4 - (b_1+b_2+b_3+b_4-1-b_1b_4-b_2b_3)^2,
\end{equation}
which serves as an effective tool to detect unistochasticity
of bistochastic matrices of size $N=3$.
If  ${\cal Q}(B)\ge 0$, the matrix $B$  is unistochastic \cite{jarlskog},
as the triangle inequality (\ref{eq:triangle}) is satisfied
and there exists a corresponding unitary matrix  $V$
such that $B=V \odot {\bar V}$.
Conversely, if  ${\cal Q} (B)$ is negative,
the unitarity triangle cannot be constructed and $B$ is not unistochastic.
This is a consequence of the fact that
the triangle conditions necessary to find two orthogonal vectors
forming first two columns of $V$
are not satisfied \cite{volume}.

Quantity (\ref{eq:jarlskog})
can be written as ${\cal Q}= 4{\mathcal J}^2$,
where ${\mathcal J}(B)={\mathcal J}(B(V))$,
 is called the \emph{Jarlskog invariant},
originally invented for the corresponding  unitary $V$  \cite{jarlskog}.
The name of this quantity is related to the following invariance property:
For two unitary matrices $V$ and $V'$,
equivalent with respect to permutations and multiplication
 from left and right by two diagonal unitary matrices,
the value of  ${\mathcal J}$ computed for the
associated bistochastic matrices $B$ and $B'$ is constant.

Interestingly, the maximal value of $Q$, equal to $1/27$,
is attained by the flat bistochastic matrix, $B_{ij}=1/3$,
which forms the center of the Birkhoff polytope ${\cal B}_3^C$.
The minimal value reads, $\mathcal{Q}_{\rm min}=-1/16$,
 and corresponds to the combination
 of two permutation matrices  $B_S=(X_3+X_3^2)/2$,
  introduced in Section  \ref{bisto},
 which is most distant from the set   $ {\cal U}_3^C$
 of unistochastic matrices  \cite{volume}.

\section{Not every transition matrix corresponding to unistochastic channel
 is unistochastic}
 \label{AppB}
Consider the unitary matrix $U\in U(9)$ defined by
\[
  U = \mathbb{I}_2\oplus H_2\oplus
  \adiag(1,-1,1)\oplus\mathbb{I}_2,
\]
where $H_2$ is the Hadamard matrix representing the Hadamard
gate, and $\adiag(1,-1,1)$ is the matrix with non-zero
entries in the antidiagonal.
Such a two-qutrit unitary matrix $U$ defines by Eq.   (\ref{eq:defuni})
a single-qutrit unistochastic channel $\Psi_U$.
The corresponding classical transition matrix
generated by hyper-decoherence, $T={\cal D}_h(\Phi_U)$, reads
\[
T =
\frac{1}{6}
\begin{pmatrix}
  5 & 1 & 0 \\
  1 & 3 & 2 \\
  0 & 2 & 4
\end{pmatrix}.
\]

Expression \eqref{eq:jarlskog} implies that its squared Jarlskog
invariant is negative,   ${\cal Q}(T)= -\frac{1}{324}$.
Therefore, the matrix $U\in U(9)$ determines
a single-qutrit unistochastic channel $\Psi_U$,
such that the corresponding classical transition matrix $T(U)$ of size $N=3$
is not unistochastic. This example shows that the set ${\cal U}_3^Q$ of
unistochastic channels (yellow region in Fig. \ref{fig:singlepanel}a)
is not included inside the blue set of
isometric unitary channels ${\cal I}_3^Q$
 which decohere to unistochastic transition matrices  \cite{KCPZ18}.
 
\section{Unistochastic channels at the corners of the star of David.}

In this appendix we present  six unitary
matrices of order $N^2=9$
which produce unistochastic channels
represented in  Fig. \ref{fig:singlepanel}a
by points at six corners of the Star of
David for the $2$-face of $\Delta_8$.
Additionally, we will show the relation between the
unistochastic channels corresponding to faces spanned
by the triads  $(\Phi_{\mathbb{I}}, \Phi_{Z},\Phi_{Z^2})$
and $(\Phi_{\mathbb{I}}$, $\Phi_{X},\Phi_{X^2})$.
\par
First, let us introduce a diagonal unitary matrix $\matrizesp \in U(9)$,
\begin{equation}
  \matrizesp
=
\mathrm{diag}(1,1,1,1,1, 1,\omega,1,\omega^2),
\label{app:unitary}
\end{equation}
where $\omega=e^{2\pi i /3}$.
Using definition Eq.~\eqref{eq:defuni}, we obtain
the Choi matrix for the corresponding unistochastic channel,
\[
  D_{\matrizesp} =
  \frac{1}{3}
  (\matrizesp^{R})^{\dagger} \matrizesp^{R}
  =
    \begin{pmatrix}
  1 & 0 & 0 & 0 & (2+\omega^2)/3  & 0 & 0 & 0 &
   (2+\omega)/3 \\
 0 & 0 & 0 & 0 & 0 & 0 & 0 & 0 & 0
   \\
 0 & 0 & 0 & 0 & 0 & 0 & 0 & 0 & 0
   \\
 0 & 0 & 0 & 0 & 0 & 0 & 0 & 0 & 0
   \\
   (2+\omega)/3 & 0 & 0 & 0 & 1
   & 0 & 0 & 0 & (2+\omega^2)/3 \\
 0 & 0 & 0 & 0 & 0 & 0 & 0 & 0 & 0
   \\
 0 & 0 & 0 & 0 & 0 & 0 & 0 & 0 & 0
   \\
 0 & 0 & 0 & 0 & 0 & 0 & 0 & 0 & 0
   \\
 (2+\omega^2)/3 & 0 & 0 & 0 &
   (2+\omega)/3  & 0 & 0 & 0 & 1
   \\
  \end{pmatrix},
\]
which can be represented as
$D_{\matrizesp}=\bigl(\frac{2}{3}\Phi_{\mathbb{I}} +
\frac{1}{3}\Phi_{Z}\bigr)^{R}$.
This implies that the unistochastic channel
 $\Phi_{U_Z}=\frac{2}{3}\Phi_{\mathbb{I}} + \frac{1}{3}\Phi_{Z}$,
 so that it is represented by a point at  one of the six
corners of the star of David inscribed in the triangle
$\Delta(\Phi_\mathbb{I}, \Phi_Z, \Phi_{Z^2})$.
In a similar way one can construct unistochastic
channels corresponding to the remaining corners of the star.

Let us now consider the other face of the simplex of Weyl channels
 determined by $(\Phi_{\mathbb{I}}$, $\Phi_{X},\Phi_{X^2})$.
The channel
$\frac{1}{3}\Phi_{\mathbb{I}}+\frac{2}{3}\Phi_{X}$
at a corner of the star of David inscribed into this triangle
and plotted in  Fig. \ref{fig:singlepanel}a,
corresponds to the unitary matrix
$\mathbb{F}\overline{\matrizesp}\overline{\mathbb{F}}$.
we use here  the tensor product of two Fourier matrices, $\mathbb{F} = F_3\otimes F_3$,
while $\overline{(\cdot)}$ denotes the complex conjugation.
\par
Combining the tensor product  of matrix $Z$ of order three with Fourier matrices and
matrix $\matrizesp$ introduced above, we obtain the remaining channels
at six corners of the star of David at this face,
listed in
Tab.~\ref{tab:unitariesNcorners}.
\begin{table}[h]
  \centering
 \begin{tabular}{ll}
   Unitary matrix $U \in U(9)$ \ \ \ \ \ &  Unistochastic channel  $\Phi_U=[(U^R)^{\dagger}U^R]^R$ \\
   \hline
 $\mathbb{F}\matrizesp\overline{\mathbb{F}}$ &$\frac{2}{3}\Phi_{\mathbb{I}}+\frac{1}{3}\Phi_{X}$\\
 $\mathbb{F}\overline{\matrizesp}\overline{\mathbb{F}}$ & $\frac{2}{3}\Phi_{\mathbb{I}}+\frac{1}{3}\Phi_{X^2}$\\
 $\mathbb{F}(Z\otimes Z)\matrizesp\overline{\mathbb{F}}$ & $\frac{1}{3}\Phi_{\mathbb{I}}+\frac{2}{3}\Phi_{X^2}$\\
 $\mathbb{F}(Z\otimes Z)\overline \matrizesp\overline{\mathbb{F}}$& $\frac{1}{3}\Phi_{X}+\frac{2}{3}\Phi_{X^2}$\\
 $\mathbb{F}(Z^2\otimes Z^2) \matrizesp\overline{\mathbb{F}}$ & $\frac{1}{3}\Phi_{X^2}+\frac{2}{3}\Phi_{X}$\\
 $\mathbb{F}(Z^2\otimes Z^2)\overline \matrizesp\overline{\mathbb{F}}$ & $\frac{1}{3}\Phi_{\mathbb{I}}+\frac{2}{3}\Phi_{X}$\\
 \end{tabular}
 \caption{
   Unitary matrix $U$ of size $9$ (left) determining by \eqref{eq:defuni} the $N=3$ unistochastic
    channel $\Phi_U$ (right) corresponding to a corner of
   the Star of David in Fig.   \ref{fig:singlepanel}a.
   Matrix $\matrizesp$ is defined in Eq.~\eqref{app:unitary}, while
   $\mathbb{F} = F_3\otimes F_3$, where  $F_3$ denotes the
   Fourier matrix of order three.
 }
 \label{tab:unitariesNcorners}
\end{table}
\par

As demonstrated in Table~\ref{tab:unitariesNcorners}
all the corners of the star correspond to unistochastic channels.
Numerical results support the conjecture that all points inside the
star also represent unistochastic channels.
Not being able to prove this statement so far, we can show
that the problem of finding the boundary of the
sets of unistochastic channels in the triangles
$\Delta(\Phi_{\mathbb{I}}, \Phi_{Z}, \Phi_{Z^2})$ and
 $\Delta(\Phi_{\mathbb{I}}, \Phi_{X}, \Phi_{X^2})$
are equivalent.

\begin{proposition}
Consider a unistochastic channel
$\Phi_p =
  p_1
  \Phi_{\mathbb{I}}
  +
  p_2
  \Phi_{Z}+
  p_3
  \Phi_{Z^2}
$,
associated with an unitary matrix $U$ of order nine.
Then the local transformation of $U'=(F_3\otimes
F_3)U(\overline{F}_3\otimes\overline{F}_3)$, corresponds to
the channel
$ \Phi'_p= p_1  \Phi_{\mathbb{I}}+  p_2  \Phi_X+ p_3  \Phi_{X^2}$
belonging to the triangle
$\Delta(\Phi_{\mathbb{I}}, \Phi_{X}, \Phi_{X^2})$.

\end{proposition}
\begin{proof}
We consider a unistochastic channel
 $\Phi_p =  p_1    \Phi_{\mathbb{I}}+
  p_2   \Phi_Z+  p_3   \Phi_{Z^2}$,
with associated unitary matrix  $V$.
 Our aim is to show that there exists a
local transformation that transforms it into $\Phi'_p$.
In other words, we show that
a unistochastic channel in the cross-section $(\mathbb{I}, Z, Z^2)$
can be transformed into a unistochastic channel in the cross-section
$(\mathbb{I}, X, X^2)$ by means of a local transformation.
\par
Firstly, we find the expression for the channel entries
in four-index notation:
\begin{align*}
  \bra{e_mf_{\mu}} \Phi \ket{e_n f_{\nu}} &=
\delta_{mn} \delta_{\mu\nu}
(
  p_1  + p_2 \omega^{n-\nu}
  + p_3 \omega^{2(n-\nu)}
) =
\Phi_{\sub{m\nu\\ n\nu}}
\end{align*}
Then the corresponding reshuffling is
\[
  (\Phi^{R}_{\sub{mn\\\mu\nu}}) =
  \delta_{m\mu}\delta_{n\nu}
( p_1  + p_2 \omega^{\mu-\nu}
  + p_3 \omega^{2(\mu-\nu)}).
\]
On the other hand, Choi matrix entries corresponding to a
unitary matrix $V$ are
\[
  (\Phi_V)_{\sub{s\sigma\\t\tau}} =
  \frac{1}{3} \sum_{l\lambda}
  V_{\sub{sl\\\sigma\lambda}}
  \overline V_{\sub{tl\\\tau\lambda}}.
\]
\par
Thus the relation between $V$ entries and $\vec p =
(p_1,p_2,p_3)$ is
\begin{equation}
  \frac{1}{3} \sum_{l\lambda}
  V_{\sub{ml\\n\lambda}}
  \overline V_{\sub{\mu l\\\nu\lambda}}
  =
  \delta_{m\mu}
  \delta_{n\nu}
  ( p_1  + p_2 \omega^{\mu-\nu}
    + p_3 \omega^{2(\mu-\nu)}).
    \label{eq:equality}
\end{equation}
\par
Now we proceed to find the entries of the matrix
$
  (F_3\otimes F_3)
  U
  (\overline F_3\otimes \overline F_3)
$,
\begin{align}
  \bra{e_mf_{\mu}}
  (F_3\otimes F_3)
  U
  (\overline F_3\otimes \overline F_3)
  \ket{e_nf_{\nu}}
  &=
  \sum_{l\lambda}
  F_{ml} F_{\mu\lambda}
  \bra{e_lf_{\lambda}}
  U
  (\overline F_3\otimes \overline F_3)
  \ket{e_nf_{\nu}} \nonumber\\
  &=
  \sum_{\sub{l\lambda\\t\tau}}
  F_{ml} F_{\mu\lambda}
  U_{\sub{l\lambda\\t\tau}}\overline F_{tn}\overline
  F_{\tau\nu} \nonumber\\
  &=
  \frac{1}{3^2}
  \sum_{\sub{l\lambda\\t\tau}}
  \omega_3^{ml}
  \omega_3^{\mu\lambda}
  U_{\sub{l\lambda\\t\tau}}
  \omega_3^{-tn}\omega_3^{-\tau\nu}
  = V_{\sub{m\mu\\n\nu}}.
  \label{ec:post}
\end{align}
\par
What follows now is to substitute Eq.~\eqref{ec:post}---$V =
(F_3\otimes F_3)U(\overline{F}_3\otimes \overline{F}_3)$---in the
Choi matrix definition of a
unistochastic channel, Eq.~\eqref{eq:defuni},
\begin{align}
 D_{\sub{st\\ \sigma\tau}}
=  (\Phi_V)_{\sub{s\sigma\\t\tau}}
 =
  \frac{1}{3} \sum_{l\lambda}
  \frac{1}{3} \sum_{l\lambda}
  V_{\sub{sl\\\sigma\lambda}}
  \overline V_{\sub{tl\\\tau\lambda}} &=
  \frac{1}{3}\frac{1}{3^{4}} \sum_{l\lambda}
  \bigl(
    \sum_{\sub{d_1d_2\\d_3d_4}}
    \omega_3^{sd_1}
    \omega_3^{ld_2}
    \omega_3^{-\sigma d_3}
    \omega_3^{-\lambda d_4}
    U_{\sub{d_1d_2\\d_3d_4}}
  \bigr)
                                    \\&\qquad\times
  \bigl(
    \sum_{\sub{c_1c_2\\c_3c_4}}
    \omega_3^{-c_1t}
    \omega_3^{-c_2l}
    \omega_3^{c_3\tau}
    \omega_3^{c_4\lambda}
    \overline U_{\sub{c_1c_2\\c_3c_4}}
    \bigr)  \nonumber
    \\
                                      &=
                    \frac{1}{3^{5}}
    \sum_{\sub{d_1d_2\\d_3d_4}}
    \sum_{\sub{c_1c_2\\c_3c_4}}
    \omega_3^{d_1s-d_3\sigma-c_1t+c_3\tau}
    U_{\sub{d_1d_2\\t\tau}}
    \overline U_{\sub{c_1c_2\\c_3c_4}}
    \sum_{l\lambda}
    \omega_3^{d_2l-d_4\lambda -c_2l+c_4\lambda}.
           \nonumber
    \label{eq-appendix-corners-2}
\end{align}
Making use of the following identity,
$\sum_{l\lambda}
\omega_3^{d_2l-d_4\lambda -c_2l+c_4\lambda}
= 3^2 \delta_{d_2c_2} \delta_{d_4c_4}$,
the above formulae
 can be reduced,
\begin{align}
   D_{\sub{st\\ \sigma\tau}}
   &=
                     \frac{1}{3^{5}}
   \sum_{l\lambda}
  \bigl(
    \sum_{\sub{d_1d_2\\d_3d_4}}
    \omega_3^{sd_1}
    \omega_3^{ld_2}
    \omega_3^{-\sigma d_3}
    \omega_3^{-\lambda d_4}
    U_{\sub{d_1d_2\\d_3d_4}}
  \bigr)
                                    \\&\qquad\times
  \bigl(
    \sum_{\sub{c_1c_2\\c_3c_4}}
    \omega_3^{-c_1t}
    \omega_3^{-c_2l}
    \omega_3^{c_3\tau}
    \omega_3^{c_4\lambda}
    \overline U_{\sub{c_1c_2\\c_3c_4}}
    \bigr)  \nonumber
                                    \\&=
    \frac{1}{3^3}
    \sum_{\sub{d_1d_2\\d_3d_4}}
    \sum_{\sub{c_1\\c_3}}
    \omega_3^{d_1s-d_3\sigma-c_1t+c_3\tau}
    U_{\sub{d_1d_2\\d_3d_4}}
    \overline U_{\sub{c_1d_2\\c_3d_4}}, \nonumber
\end{align}
Now using Eq.~\eqref{eq:equality} we can rewrite the entries of
the Choi matrix $D$ corresponding to the map $\Phi_V$
\begin{align*}
 D_{\sub{st\\ \sigma\tau}}
&=
\frac{1}{3^2}
\sum_{\sub{d_1c_1\\d_3c_3}}
\omega_3^{d_1s-d_3\sigma-c_1t+c_3\tau}
\delta_{d_1c_1}
\delta_{d_3c_3}(p_1+p_2\omega_3^{c_1-c_3}+p_3\omega_3^{2(c_1-c_3)})\\
  &=
  \frac{1}{3^2}
  \sum_{\sub{x\\y}}\omega_3^{xs-y\sigma-xt+y\tau}(p_1+p_2\omega^{x-y}+p_3\omega^{2(x-y)}).
\end{align*}
The last equation corresponds---after reshuffling---to the
desired channel
$\Phi_V=\Phi'_p=  p_1 \Phi_{\mathbb{I}} +
  p_2
\Phi_X+
  p_3
  \Phi_{X^2}$.
\end{proof}
In general,  for any pair of triangles,
$\Delta(\Phi_{\mathbb{I}},  \Phi_\mu, \Phi_{\bar\mu})$
and $\Delta(\Phi_{\mathbb{I}},  \Phi_\nu, \Phi_{\bar\nu})$
with  $\nu\neq\mu$,
there exists a similar relationship
determined by unitary matrices belonging to the Clifford group~\cite{Tolar_2018}.
The proof of this is similar to the one given above.

\end{document}